\documentclass[aps,superscriptaddress,twocolumn, nofootinbib, longbibliography, pra, 10pt]{revtex4-2}

\newcommand{\tuwien}{Technische Universität Wien, Austria}

\newcommand{\jkulinz}{Johannes Kepler Universität Linz, Austria}

\newcommand{\leicester}{University of Leicester, United Kingdom}

\usepackage{url}
\usepackage{physics}
\usepackage{amsmath}
\usepackage{amsthm}
\usepackage{amsfonts}
\usepackage{svg}
\usepackage{hyperref}
\usepackage{cleveref}
\usepackage{subcaption}
\usepackage{booktabs}
\usepackage{thmtools}

\newtheorem{lemma}{Lemma}
\newtheorem{theorem}{Theorem}

\begin{document}

\title{An Online Approach for Entanglement Verification Using Classical Shadows}

\author{Marwa Marso}
\thanks{Both authors contributed equally to this research. Corresponding authors: maroua.marso@jku.at, sabrina.herbst@tuwien.ac.at}
\affiliation{\jkulinz}

\author{Sabrina Herbst}
\thanks{Both authors contributed equally to this research. Corresponding authors: maroua.marso@jku.at, sabrina.herbst@tuwien.ac.at}
\affiliation{\tuwien}

\author{Jadwiga Wilkens}
\affiliation{\jkulinz}

\author{Vincenzo De Maio}
\affiliation{\tuwien}
\affiliation{\leicester}

\author{Ivona Brandi\'c}
\affiliation{\tuwien}

\author{Richard Kueng}
\affiliation{\jkulinz}

\begin{abstract} 
Quantum measurements are slow, while classical processors are fast, yet existing hybrid protocols never exploit this asymmetry. In this work, we propose an alternative formulation of classical estimators as online algorithms that are updated incrementally upon obtaining a new sample. Classical shadows are the natural fit for this approach: designed around the principle of measuring first and asking questions later, each snapshot is a self-contained classical description that can be processed immediately and independently. As a first demonstration, we focus on mixed state entanglement verification via PT-moments, moments of the partially transposed density matrix that provide experimentally accessible sufficient conditions for entanglement. We construct two unbiased online estimators that together characterize the fundamental challenge between memory footprint and per-shot computational cost: one scales to large systems at low moment order, the other handles high moment orders at the expense of memory exponential in system size. The online estimator certifies entanglement reliably and, by exploiting all $\binom{T}{m}$ combinations of snapshots, requires fewer samples than state-of-the-art baselines, turning entanglement detection from a purely offline diagnostic into a protocol that runs concurrently with the experiment.
\end{abstract}

\maketitle

\section{Introduction}

Every quantum experiment has a hidden classical resource: the time between shots, which has not been considered in existing protocols.
This is particularly critical in the noisy intermediate-scale quantum (NISQ) regime~\cite{Preskill2018quantumcomputingin}, where preparing quantum states and performing measurements is expensive, due to quantum operations being orders of magnitude slower than their classical counterparts~\cite{hoefler2023disentangling} and device overhead adding further cost~\cite{supremacy}. 
Classical processors, by contrast, are faster and can typically be parallelized. 
Between shots, the quantum device must reset, creating an interval of unavoidable slack time during which the classical processor sits idle. 
Thus, in the NISQ regime, the bottleneck is sampling, not classical post-processing.

Most existing protocols, however, treat the quantum and classical stages as two separate phases: first acquire all measurement data, then reconstruct the desired quantities offline. 
This paradigm scales poorly. 
Classical workloads accumulate toward the end of the experiments, where all measurement data needs to be considered at once, which constitutes a significant challenge for the memory subsystem.
More importantly, this approach prevents overlapping classical computation with quantum execution, leaving the classical processor idle during the data collection process.

A more natural strategy would be to exploit the reset-latency window for online classical processing already. Instead of storing all outcomes and processing them only afterward, one can update estimators incrementally as each new measurement arrives. 

This leads to a hybrid quantum-classical \emph{online} paradigm: the quantum device continuously produces measurement outcomes, while the classical computer immediately incorporates them into running estimates. 
In the best case, this yields a final estimate as soon as sampling ends, as illustrated in 
\Cref{fig:protocol-overview}.

We reformulate the traditional post-processing algorithm into an iterative procedure, allowing to calculate shadow estimates on-the-fly.
In this paper, we consider this as an online algorithm~\cite{manasse1988competitive}, in which the emphasis is on making real-time updates rather than on computing approximations over massive, non-storable datasets.  

The randomized measurement toolbox~\cite{elben2023randomized} is particularly well suited for such a framework. 
By performing randomized measurements of a target quantum state and storing the results, one can efficiently predict many state properties classically without reconstructing the full density matrix. 
Here we focus on classical shadow protocols~\cite{huang-predicting, paini2019approximatedescriptionquantumstates, morris2020selectivequantumstatetomography}, which rely on a simple, fixed, single-copy measurement primitive. 
Each snapshot is converted into a succinct classical description that reproduces the underlying state in expectation, allowing one to predict a broad class of observables and nonlinear quantities directly from the snapshots. 
Crucially, each snapshot is processed independently, keeping updates lightweight and memory overhead bounded, which are ideal properties for an online regime. 

\begin{figure*}
    \centering
    \includegraphics[width=1\linewidth]{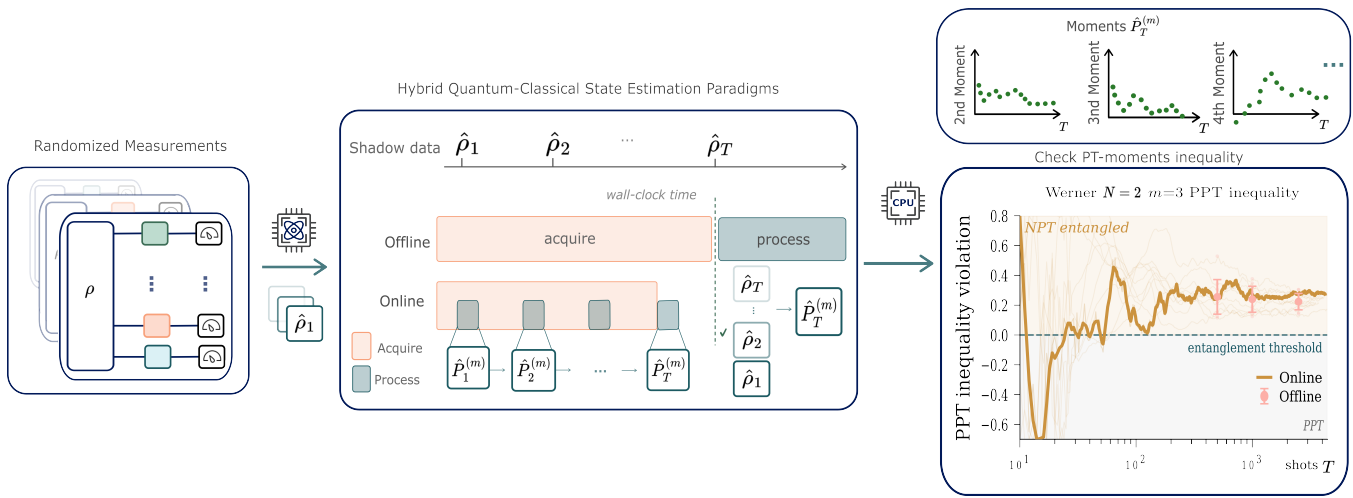}
    \caption{\textbf{From Randomized Measurements to Online Entanglement Detection.} (\textit{Left}) Each shot produces a snapshot $\hat{\rho}_i$ by applying a random Clifford rotation and measuring in the computational basis. (\textit{Middle}) In the offline setting, all $T$ snapshots are acquired before any processing begins. In the online setting, running PT-moment estimates $\hat{P}_t^{(m)}$ are updated incrementally as each new shot arrives.
    (\textit{Right}) The $m{=}3$ PT-moment inequality gap for a $2$-qubit Werner state ($t = 5/6$), tracked by the online estimator over shots. The trace crosses zero, yielding $e_3 < 0$ and certifying NPT entanglement on-the-fly.}
    \label{fig:protocol-overview}
\end{figure*}

Entanglement detection provides a compelling first target, as characterizing the entanglement properties of a quantum state is essential for many quantum information processing task, while being known to be a daunting task on real quantum hardware, where states become mixed due to noise corruption.
The partial transpose (PT) spectrum provides a powerful entanglement condition~\cite{Peres1996, Horodecki1996} that remains valid for (highly) mixed states, however, accessing it experimentally requires full state tomography. 
A practical alternative is to work with PT-moments $\mathrm{Tr}[(\rho^{T_B})^m]$, which carry sufficient information to detect entanglement via hierarchies of inequalities~\cite{moments-ppt, moments-ppt-2, 10.21468/SciPostPhys.12.3.106} and can be estimated from single-copy randomized measurements without reconstructing the full state. 
Since PT-moments are nonlinear functionals of the state, however, their unbiased estimation carries an inherent combinatorial cost that grows with both the moment order and the number of shots, which is a complexity intrinsic to the problem. 
Online processing does not remove this cost, but it allows the classical computation to proceed concurrently with the experiment, turning PT-moment estimation from a purely offline diagnostic into a tool that can be evaluated on-the-fly.

Here, we express PT-moments within the classical shadow framework as unbiased estimators written as polynomial functions of snapshots from single-copy measurements, and show how to implement them as an online algorithm that process each snapshot concurrently with the experiment, thus, exploiting the slack time between shots to perform classical computation. 
Our contribution is both algorithmic and conceptual. 
On the algorithmic side, we propose two online estimators that together characterize the fundamental and unavoidable choice between memory footprint and per-shot computational cost of classical shadow applications. 
The first estimator operates directly on the raw measurement record and scales to large systems at low moment order, while the other maintains a fixed set of accumulated matrices with constant per-shot cost, making it the better fit when high moment orders are required. 
Both estimators are unbiased by construction and achieve better statistical performance than existing approaches in the practically relevant regime~\cite{batched-shadows}. 
On the conceptual side, we argue that the online paradigm is not merely a computational optimization but the appropriate architectural model for hybrid quantum-classical computation: classical computation should be interleaved with measurement acquisition. 
The classical shadow framework, due to its incremental structure, naturally supports this model.

The paper is structured as follows. \Cref{sec:preliminaries} presents preliminaries on classical shadows and entanglement detection via PT-moments. 
In \Cref{sec:pt-moment-strategies}, we review existing strategies for PT-moment estimation before introducing our online approach in \Cref{sec:streaming-algorithm}. 
Experimental results and comparisons to existing approaches appear in \Cref{sec:results}, followed by conclusions 
in \Cref{sec:discussion-conclusion}.

\section{Preliminaries}\label{sec:preliminaries}
\subsection{The Classical Shadows Formalism}

Consider an unknown quantum state \(\rho\) on \(N\) qubits whose properties we want to predict. 
A succinct classical representation can be obtained using the classical shadow framework~\cite{huang-predicting,paini2019approximate,morris2019selective}. 
The framework applies to any ensemble of unitaries $\mathcal{U}$, however, here we focus on the local random Pauli setting, where we rotate the state according to the tensor product of $N$ randomly-drawn unitaries $U_i \in \{\mathbb{I}, H, HS\}$ from a subset of the single-qubit Clifford group and measure in the computational basis. 
This process corresponds to random Pauli-basis measurements on each qubit and is visualized in Figure~\ref{fig:shadows-overview} (left). 

\begin{figure*}
    \centering
    \includegraphics[width=0.75\linewidth]{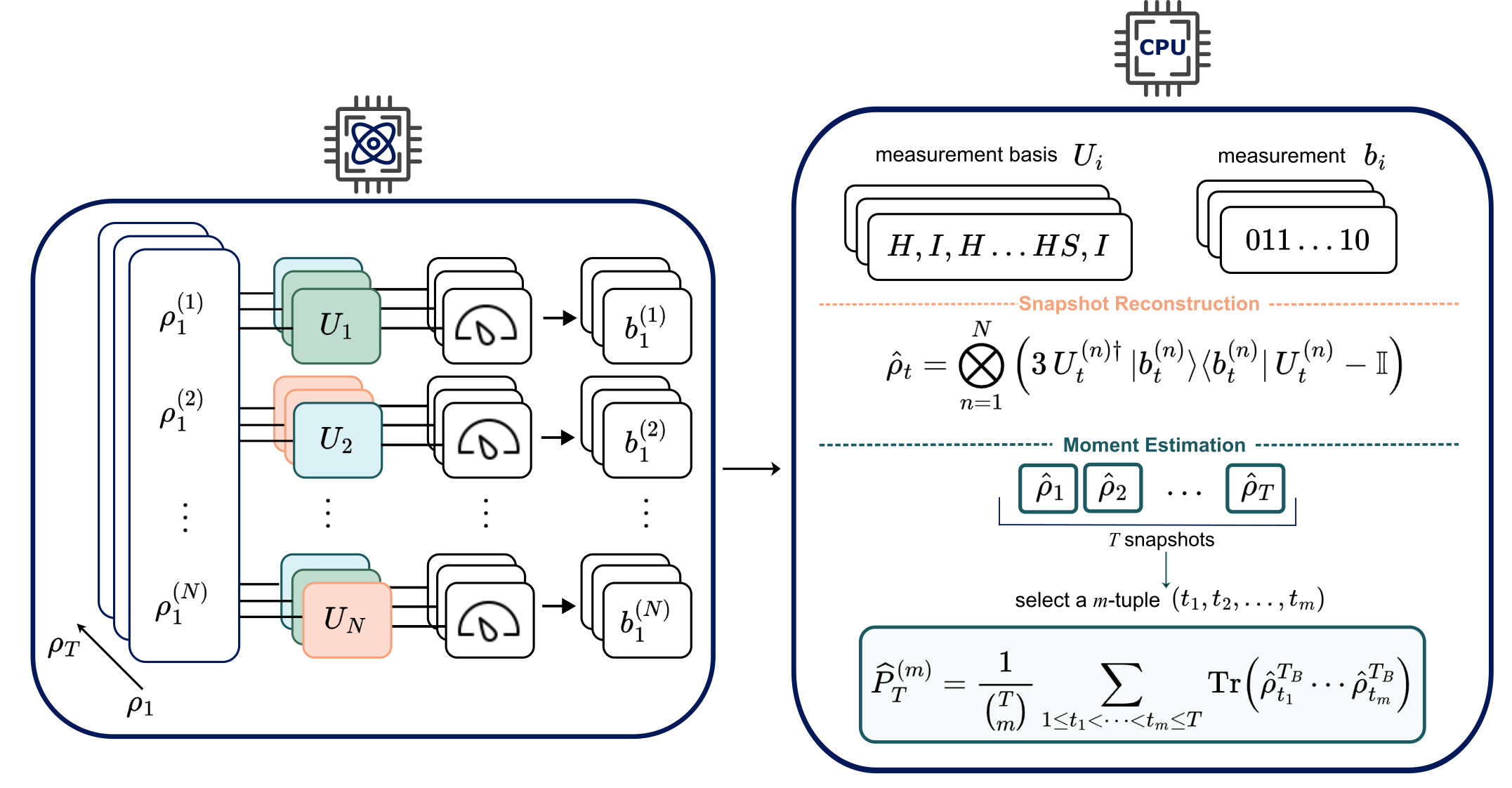}
    \caption{\textbf{Randomized measurements and classical-shadow reconstruction.} For each shot $t \in [T]$, we prepare a copy of the state $\rho$, apply the tensor product of random single-qubit Clifford gates $U_t = \bigotimes_{n=1}^{N} U_n$, measure in the computational basis and return the measurement bit string $b_t \in \{0,1\}^N$. We can then construct the single-shot classical shadow $\hat{\rho}_t$ from the pair $(U_t,b_t)$, which can in turn be used to estimate linear observables as well as higher-order functions, such as the PT-moments, by combining $m$-tuple snapshots into the unbiased estimator of $\mathrm{Tr}(\rho^{T_B})^m$.
    For the purity, a quadratic function, this formula simplifies to $\hat{P}_T^{(2)}= \frac{2}{T(T+1)} \sum_{t_1 \neq t_2} \mathrm{Tr} \left( \rho_{t_1} \rho_{t_2} \right)$.}
    \label{fig:shadows-overview}
\end{figure*}

For each shot, the measurement outcome \(b \in \{0,1\}^N\) is stored together with the chosen unitary \(U = \otimes_{n\in [1\dots N]}U_n\). This outcome occurs with probability $\Pr(b) = \bra{b} U \rho U^\dagger \ket{b}$, and corresponds to the post-measurement state $U^\dagger \ketbra{b}{b} U$. In expectation over the measurement randomness, this defines a quantum channel
\begin{equation}
    \mathcal{M}(\rho) := \mathbb{E}_{U,b}\bigl[ U^\dagger \ketbra{b}{b} U \bigr],
\end{equation}
which, viewed as a linear map, has a unique inverse $\mathcal{M}^{-1}$ that can be computed analytically. 
A single classical snapshot is then given by
\begin{equation}
    \hat{\rho} = \mathcal{M}^{-1}\!\bigl( U^\dagger \ketbra{b}{b} U \bigr).
\end{equation}
By construction, $\mathbb{E}_{U,b}[\hat{\rho}] = \rho$, so each snapshot is an unbiased estimate of $\rho$ and any quantity computed from empirical averages converges to the true value. Note, however, that the individual snapshots are unphysical in the sense that they feature negative eigenvalues. 
In the local random Pauli setting, the snapshot admits a convenient tensor product representation across qubits,

\begin{equation}
    \hat{\rho} = \bigotimes_{n=1}^N \hat{\rho}_{n}, 
    \qquad 
    \hat{\rho}_{n} = 3\, U_n^\dagger \ketbra{b_n}{b_n} U_n - \mathbb{I}.
    \label{single_shot_shadow}
\end{equation}

After performing $T$ shots, we define the classical shadow of $\rho$ as the collection of single snapshots
\begin{equation}
    \begin{split}
        S(\rho, T) = \big\{ & \hat{\rho}_1 = \mathcal M^{-1}\big(U_1^\dagger\ket{b_1}\!\bra{b_1}U_1\big), \dots, \\ &  \hat{\rho}_T  = \mathcal M^{-1}\big(U_T^\dagger\ket{b_T}\!\bra{b_T}U_T\big)\big\}.
    \end{split}
\end{equation}
Averaging over the $T$ snapshots yields the empirical estimator $\hat{o} = \frac{1}{T}\sum_{t=1}^T \mathrm{Tr}(O\hat{\rho}_t)$, which is unbiased since linearity of the trace gives $\mathbb{E}[\hat{o}] = \mathrm{Tr}(O\rho)$. 
More generally, nonlinear functions of $\rho$ can be estimated by forming products of independent snapshots, which opens the door to entanglement detection via quantities such as the PT-moments introduced in the next section. This classical data processing stage is illustrated in Figure~\ref{fig:shadows-overview} (right).

\subsection{Entanglement detection using trace moments of the partially transposed density matrix (PT moments)}

Let $\rho_{AB}$ be a bipartite density matrix and take the partial transpose on subsystem $B$, denoted $\rho^{T_B}$.
For integer $m \ge 1$, the $m$-th partially transposed trace moment, abbreviated as PT-moment, is given by 
\begin{equation}
    P^{(m)} = \mathrm{Tr}\!\left[(\rho^{T_B})^{m}\right]
    \label{eq:state-moments}
\end{equation}

Since $\rho^{T_B}$ is Hermitian with real eigenvalues ${\lambda_i}$, the PT-moments coincide with the power sums $P^{(m)} = \sum_i \lambda_i^m$. 
In particular, $P^{(1)} = \mathrm{Tr}(\rho^{T_B}) = 1$ and $P^{(2)} = \mathrm{Tr}(\rho^2)$ since partial transposition preserves the trace and the purity. 
For higher moments $m \geq 3$, $P^{(m)}$ can differ substantially from $\mathrm{Tr}(\rho^m)$ if $\rho^{T_B}$ has negative eigenvalues, which, by the PPT condition~\cite{Peres1996, Horodecki1996}, certifies entanglement. 
The PT-moments thus carry information about entanglement, which can be extracted via the \emph{elementary symmetric polynomials} (ESPs) of the eigenvalues of $\rho^{T_B}$,
\begin{equation}
e_k := \sum_{1 \le i_1 < \cdots < i_k \le d} \lambda_{i_1} \cdots \lambda_{i_k}, \qquad k = 1, \dots, d,
\label{eq:esp-def}
\end{equation}
which are computable from the PT-moments via the Newton--Girard recurrence~\cite{moments-ppt-2}
\begin{equation}
k\,e_k = \sum_{j=1}^{k}(-1)^{j-1}\, P^{(j)}\, e_{k-j}, \qquad e_0 := 1.
\label{eq:newton-girard}
\end{equation}
If $\rho^{T_B} \succeq 0$, all eigenvalues are nonnegative and $e_k \ge 0$ for all $k$. Conversely, $e_k < 0$ for any $k$ certifies that $\rho^{T_B}$ has at least one negative eigenvalue and hence $\rho$ is NPT entangled. The conditions ${e_k \ge 0}$ thus form a hierarchy that is necessary and sufficient for the PPT condition, and a violation at any single level is a sufficient condition for entanglement. We refer to \Cref{app:ng} for the full derivation.

Classical shadows provide a concrete means to estimate the PT-moments $\{P^{(m)}\}$ via randomized measurements, making this hierarchy experimentally accessible. Such approaches have already been tested on data from experimental hardware, see e.g.\ \cite{moments-ppt,moments-ppt-2,10.21468/SciPostPhys.12.3.106} for state entanglement and Ref.~\cite{batched-shadows} for operator entanglement. We review some of the prevalent approaches, as well as their limitations in the next section. 
\section{PT-moment estimation strategies}\label{sec:pt-moment-strategies}
Our goal is to estimate the PT-moments \(P^{(m)}=\mathrm{Tr}[(\rho^{T_B})^m]\) from classical shadow data, which is known to be very data-intensive in general. 
Indeed, Refs.~\cite{9719827, 10756089} prove that estimating the second moment ($m=2$) of a $n$-qubit state from \emph{any} single-copy measurement protocol requires at least $\Omega \left(2^{n/2} \right)$ shots for certain quantum states. 
Since classical shadows are a single-copy measurement protocol, any PT-moment estimator based on them must incur the same exponential sample complexity. 
When it comes to actual estimation protocols, several offline strategies exist. These range from unbiased but computationally expensive estimators to simpler plug-in methods that introduce structural biases. We review these approaches before presenting our online estimator.

\subsection{Unbiased U-statistics}
\label{sec:pt-moments-combinatorial}

A natural estimator for the $m$-th PT-moment is the symmetric $U$-statistic~\cite{huang-predicting, moments-ppt},
\begin{equation}
    \widehat{P}^{(m)}_T
    := \frac{1}{\binom{T}{m}}
       \sum_{1 \le t_1 < \cdots < t_m \le T}
       \mathrm{Tr}\!\left(
           \hat{\rho}_{t_1}^{T_B}
           \cdots
           \hat{\rho}_{t_m}^{T_B}
       \right),
    \label{eq:pt-moment-shadow-estimator}
\end{equation}
where $\{\hat{\rho}_t\}_{t=1}^T$ are i.i.d.\ classical snapshots of $\rho$. Since the snapshots are independent and each satisfies $\mathbb{E}[\hat{\rho}_t^{T_B}] = \rho^{T_B}$, the expectation of each product factorizes and the estimator is unbiased,
\begin{equation}
 \mathbb{E}\bigl[\widehat{P}^{(m)}_T\bigr] = P^{(m)}.
\end{equation}
However, computing all $\binom{T}{m}$ snapshot products requires repeated passes over the full dataset. Because the snapshot $\hat\rho_t = \bigotimes_{n=1}^N \hat\rho_{t,n}$ factorizes across qubits, each trace product reduces to a product of $N$ single-qubit traces. We exploit this factorization later in our online estimators. 
Thus, evaluating each term has an associated computational complexity of $O(N)$, summing up to a total of $\Theta(T^m N)$. 
As a result, this approach is computationally infeasible for moderate $T$ or $m$. 

\subsection{Plug-in estimator via density matrix reconstruction}
\label{sec:plug-in}

A simpler alternative avoids the combinatorial sum entirely. One first averages all $T$ snapshots into a single reconstructed matrix~\cite{PhysRevLett.111.160406, Guta_2020},
\begin{equation}
    \bar{\rho}_T := \frac{1}{T}\sum_{t=1}^T \hat{\rho}_t,
    \label{eq:empirical-average}
\end{equation}
and then evaluates the PT-moment directly on it, defining the \emph{plug-in estimator}
\begin{equation}
    \bar P_T^{(m)} := \operatorname{Tr}\!\left[(\bar\rho_T^{T_B})^m\right].
\end{equation}
Note that in the context of quantum state tomography, Eq.~\eqref{eq:empirical-average} is known as the linear inversion estimator of the density matrix $\rho$~\cite{PhysRevLett.111.160406,Guta_2020}.

One fundamental drawback of this approach is that the estimation is biased. While $\mathbb{E}[\bar{\rho}_T] = \rho$, it does \emph{not} follow that $\mathbb{E}[\bar{P}_T^{(m)}] = P^{(m)}$, as the expectation commutes with linear operations but not with the nonlinear map $\rho \mapsto \operatorname{Tr}[(\rho^{T_B})^m]$. The resulting bias is of order $O(1/T)$ and vanishes asymptotically, but can be substantial at finite $T$, particularly at higher moment orders $m$.
In addition to the bias, storing $\bar{\rho}_T$ requires $O(4^N)$ memory for the full $2^N\times 2^N$ matrix, and evaluating the moment requires $m{-}1$ matrix multiplications each costing $O(8^N)$. Both scale exponentially with the system size, limiting this approach to small systems.

\begin{table*}[!t]
    \centering
   
    \begin{tabular}{lcccl}
    \hline
         \textbf{Method} & \textbf{Unbiased} & \textbf{Memory} & \textbf{Per-moment cost} & \textbf{Streamable} \\
         \hline
         U-statistic (\Cref{sec:pt-moments-combinatorial}) & yes & $O(TN)$ & $O(\binom{T}{m} N)$ & no \\
         Plug-in (\Cref{sec:plug-in}) & no & $O(4^N)$ & $O(m \cdot 8^N)$ & yes \\
         Batched (\Cref{sec:batched-shadows}) & yes & $O(N_B \cdot 4^N)$ & $O(\binom{N_B}{m} m \cdot 8^N)$ & no \\
         \hline
    \end{tabular}
     \caption{Comparison of PT-moment estimation strategies. $T$: number of shots, $N$: number of qubits, $m$: moment order, $N_B$: number of batches.}
    \label{tab:baselines}
\end{table*}

\subsection{Batched shadows}
\label{sec:batched-shadows}

Batched shadows~\cite{batched-shadows} offer a middle ground between the full $U$-statistic and the plug-in estimator. Instead of averaging all $T$ snapshots into a single reconstructed state, the $T$ shots are partitioned into $N_B$ disjoint batches, and, within each batch, the snapshots are averaged to form a batch estimator
\begin{equation}
    \bar{\rho}_b = \frac{N_B}{T}\sum_{j \in b} \hat{\rho}_j, \qquad b = 1, \dots, N_B.
\end{equation}
The PT-moment is then estimated by a $U$-statistic over the $N_B$ batch estimators rather than the $T$ individual snapshots:
\begin{equation}
    \widehat{P}^{(m)}_{\mathrm{batch}} = \frac{1}{\binom{N_B}{m}} \sum_{1 \le b_1 < \cdots < b_m \le N_B} \operatorname{Tr}\!\left(\bar{\rho}_{b_1}^{T_B} \cdots \bar{\rho}_{b_m}^{T_B}\right).
    \label{eq:batched-estimate}
\end{equation}
In contrast to the plug in estimator, this  procedure is unbiased and reduces the combinatorial cost from $\binom{T}{m}$ to $\binom{N_B}{m}$ terms, since $N_B \ll T$. The number of batches must satisfy $N_B \ge m$.
The price of batching is statistical efficiency: by averaging snapshots within each batch before forming cross-products, cross-batch correlations between individual snapshots are discarded. In particular, of all $\binom{T}{m}$ possible $m$-tuples of distinct snapshots, only those in which all $m$ snapshots fall into different batches contribute to the estimator. This reduces the effective number of terms and increases variance compared to the full $U$-statistic, especially at higher moment orders where the number of discarded cross-block terms grows combinatorially.
More precisely, the approach requires storing $N_B$ density matrices of size $2^N \times 2^N$, resulting in $O(N_B \cdot 4^N)$ memory and $\binom{N_B}{m}$ trace-product evaluations, each costing $O(m \cdot 8^N)$. For small $N_B$, both are manageable for moderate system sizes.

\subsection{Summary and motivation for the online approach}

\Cref{tab:baselines} compares the baseline strategies. The unbiased $U$-statistic is statistically optimal but computationally infeasible; the plug-in estimator is cheap but biased; batched shadows are unbiased and tractable but sacrifice statistical efficiency by discarding cross-block terms. All three approaches that involve state reconstruction require $O(4^N)$ memory per stored matrix.
None of these approaches is simultaneously unbiased, memory-efficient, and streamable.

This gap motivates the online estimators that we develop in \Cref{sec:streaming-algorithm}. The key idea is to restructure the $U$-statistic into an online form that processes each snapshot exactly once and is unbiased by construction. We define two complementary online algorithms: one that minimizes memory footprint for intermediate calculations (\Cref{sec:streaming-no-reconstruction}) and one that minimizes per-shot computation time (\Cref{sec:streaming-reconstructing-dm}). Both exploit the slack time between measurements.

\section{Online PT-moment estimation via hybrid quantum-classical processing}\label{sec:streaming-algorithm}
In a hybrid quantum-classical system, the quantum processor and the classical backend operate on fundamentally different timescales. Current superconducting and trapped-ion devices execute single-shot measurements at rates ranging from tens of shots per second up to a few kilohertz~\cite{gunyho2024single,PhysRevApplied.7.054020}, while a modern CPU core can perform billions of floating-point operations per second. As a consequence, between any two consecutive measurement shots, a classical processor has substantial slack time that can be exploited for computation.
Classical shadows are particularly well suited to this setting: each snapshot is independent and carries no information about the measurement history, allowing classical post-processing to begin immediately as soon as a new measurement arrives. This paves the way to online estimators that incrementally update PT-moment estimates during the experiment, rather than processing the full dataset afterward. 

None of the estimators reviewed in \Cref{sec:pt-moment-strategies} leverages this asymmetry. We now address this gap with two online estimators that navigate the inherent choice between optimizing memory footprint and per-shot computational cost, depicted in \Cref{fig:seq-moment-estimator} and detailed in \Cref{sec:streaming-no-reconstruction,sec:streaming-reconstructing-dm}.

\begin{figure*}
    \centering
    \textbf{(a)}\\[2pt]
    \includegraphics[width=0.9\linewidth]{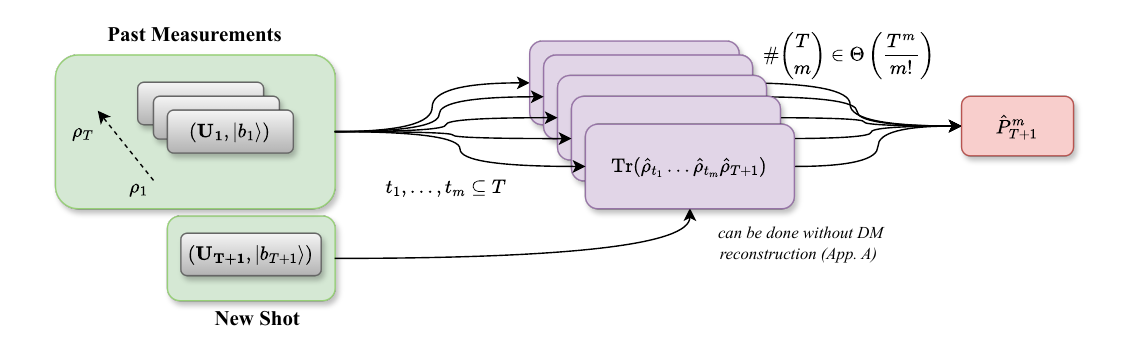}\\[10pt]
    \textbf{(b)}\\[2pt]
    \includegraphics[width=0.9\linewidth]{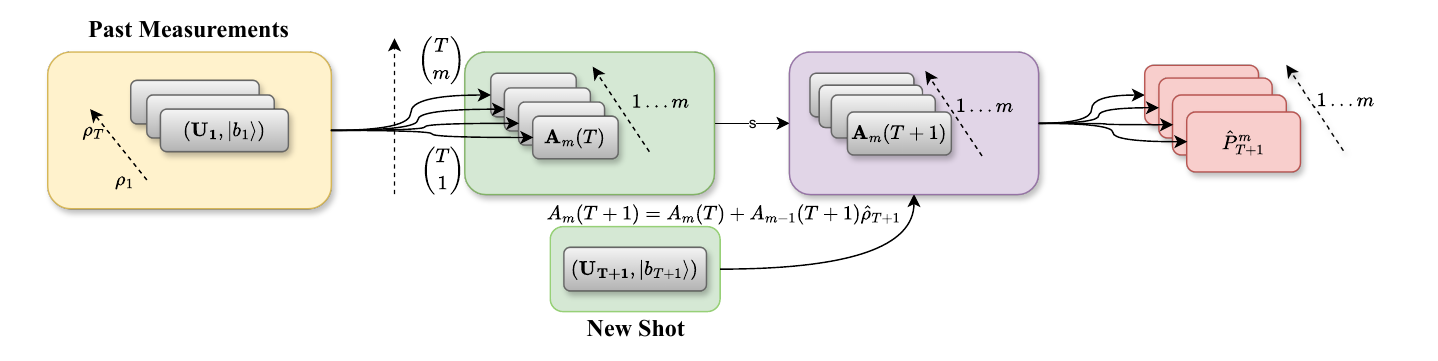}
    \caption{\textbf{Online PT-moment estimators.} 
    \textbf{(a)} Without shadow reconstruction: the full measurement record $\{(U_t, b_t)\}_{t=1}^{T}$ is kept in memory (green). Upon arrival of shot $T{+}1$, all $\binom{T}{m-1}$ subsets of past snapshots are combined with the new snapshot to compute $\binom{T}{m-1}$ trace products $\operatorname{Tr}(\hat{\rho}_{t_1}^{T_B} \cdots \hat{\rho}_{T+1}^{T_B})$ (purple), yielding the updated estimate $\hat{P}^{(m)}_{T+1}$ (red). Snapshot matrices need never be explicitly constructed.
    \textbf{(b)} With shadow reconstruction: $m$ accumulated matrices $A_1(T), \dots, A_m(T)$ are kept in memory (green). Upon arrival of shot $T{+}1$, each matrix is updated via $A_k(T{+}1) = A_k(T) + A_{k-1}(T)\,\hat{\rho}_{T+1}^{T_B}$ (purple), and the new snapshot is immediately discarded. The estimate $\hat{P}^{(m)}_{T+1}$ follows from a single trace (red).}
    \label{fig:seq-moment-estimator}
\end{figure*}

\subsection{Online algorithm without constructing empirical density matrices  (memory efficient)}\label{sec:streaming-no-reconstruction}
We now rewrite the unbiased PT-moment estimator from \Cref{eq:pt-moment-shadow-estimator} in an online fashion. 
For $T \ge m$, let $\widehat{P}^{(m)}_T$ denote the estimator of the $m$-th PT-moment after $T$ shots.
When the $(T+1)$-th shot arrives, the estimator is updated as

\begin{equation}
    \begin{split}
        \widehat{P}^{(m)}_{T+1} & =
    \Bigl(1 - \frac{m}{T+1}\Bigr)\, \widehat{P}^{(m)}_{T}
            \;+\;
    \frac{m}{T+1}\;
    \frac{1}{\binom{T}{\,m-1\,}} \\
    & \sum_{1 \le i_1 < \cdots < i_{m-1} \le T}
    \operatorname{Tr}\!\Bigl[
      \hat\rho_{i_1}^{T_B}\cdots\hat\rho_{i_{m-1}}^{T_B}\,
      \hat\rho_{T+1}^{T_B}
    \Bigr].
    \end{split}
\label{eq:streaming-update}
\end{equation}
This update rule follows from decomposing all $m$-tuples drawn from ${1,\dots,T+1}$ into those that exclude the newest shot and those that include it. The contribution of the first subset is already captured by $\widehat{P}^{(m)}T$, therefore only the $(m{-}1)$-tuples from ${1,\dots,T}$ paired with $\hat{\rho}{T+1}$ require fresh computation. As shown in \Cref{app:combinatorial-moment}, each trace product $\operatorname{Tr}[\hat{\rho}_{i_1}^{T_B}\cdots\hat{\rho}_{T+1}^{T_B}]$ can be evaluated directly from the raw measurement outcomes and Pauli axis indices, without ever constructing the $2^N\times 2^N$ snapshot matrices. The $3^m$ single-qubit Pauli trace values can be precomputed once and reused across all updates. As proved in \Cref{app:streaming-update}, the resulting estimator coincides with the $U$-statistic for all $T\ge m$ and is therefore unbiased.

These two observations have direct consequences for memory and runtime. On the memory side, only the raw measurement record needs to be retained: $T$ pairs $(b_{t,n},\sigma_{t,n})$ of bit outcomes and Pauli axis indices, giving $O(TN)$ memory which is exponentially cheaper than the $O(4^N)$ cost of any density matrix based approach. On the computational side, processing shot $T{+}1$ requires evaluating $\binom{T}{m-1}=\Theta(T^{m-1})$ trace products, a cost that grows as a degree-$(m{-}1)$ polynomial in $T$. This is tractable for small $m$ and moderate $T$, but becomes the dominant bottleneck beyond these regimes. Nevertheless, by interleaving classical computation with data acquisition, the online approach can reduce the wall-clock time even when the total work matches the offline $U$-statistic.

\subsection{Online algorithm with construction of empirical density matrices (data locality)}
\label{sec:streaming-reconstructing-dm}

The limited scalability of the approach in \Cref{sec:streaming-no-reconstruction} stems from the combinatorial update cost: processing the $(T+1)$-th snapshot requires iterating over all $\binom{T-1}{m-1}$ combinations of past shots, a $\Theta(T^{m-1})$ operation that quickly becomes the dominant runtime bottleneck. We now present a second online estimator that eliminates this dependence on $T$ entirely, by maintaining a fixed set of $m$ accumulated matrices rather than the full measurement record.
The key observation, derived in \Cref{app:streaming-update}, is that the unbiased PT-moment estimator can be rewritten as
\begin{equation}
    \widehat{P}^{(m)}_T 
    = \frac{1}{\binom{T}{m}}\,\operatorname{Tr}\!\bigl[A_{m-1}(T)\bigr],
    \label{eq:reconstruction-moment}
\end{equation}
where $A_k(T) \in \mathbb{C}^{2^N \times 2^N}$, for $k = 0, \dots, m-1$, 
accumulate ordered products of snapshots,
\begin{align}
    A_k(T) &= \sum_{1 \le i_1 < \cdots < i_k \le T} 
               \hat\rho_{i_1}^{T_B} \cdots \hat\rho_{i_k}^{T_B},
    \quad k = 1, \dots, m-1,
    \label{eq:Ak-definition}
\end{align}
Upon arrival of the $(T{+}1)$-th shot, these matrices satisfy the recursive update rule
\begin{align}
    A_0 &\leftarrow A_0 + \hat{\rho}_{T+1}^{T_B}, \nonumber \\
    A_k &\leftarrow A_k + A_{k-1}^{\,\text{prev}} \cdot \hat{\rho}_{T+1}^{T_B}, 
    \quad k = 1, \dots, m-1,
    \label{eq:reconstruction-update}
\end{align}
where $A_{k-1}^{\,\text{prev}}$ denotes the value of $A_{k-1}$ before the current update. Crucially, the snapshot $\hat\rho_{T+1}^{T_B}$ can be discarded immediately afterward: the full state of the estimator is captured by the $m$ matrices alone.
The per-shot cost is exactly $m$ matrix additions and $m-1$ matrix multiplications — both of which are independent of the shot number $T$. The total memory footprint is $O(m \cdot 4^N)$: exponential in system size but independent of $T$, and growing only linearly in moment order $m$.

\subsection{Memory vs. Data Locality}\label{sec:memory-locality}

The two online estimators from Section~\ref{sec:streaming-no-reconstruction} and Section~\ref{sec:streaming-reconstructing-dm} thus occupy complementary regimes. The reconstruction-free approach is suited to large systems where $O(4^N)$ memory is prohibitive, but is limited to moderate $T$ and $m$ due to its $\Theta(T^{m-1})$ update cost. In contrast, the reconstruction-based estimator is restricted to smaller systems, where storing $O(4^N)$-dimensional matrices is feasible, but scales to arbitrarily large $T$ and supports high moment orders $m$, making it the natural choice when entanglement detection requires higher PT-moments.

It is instructive to reflect on these tradeoffs in the broader context of feature estimation with randomized measurements.
The very formulation of classical shadows assigns equal weight to every measurement shot, which makes the online update in \Cref{eq:streaming-update} inherently data-intensive, as it immediately requires re-evaluation of the whole dataset multiple times for each new update. This works for a moderate shot number $T$ and small moment order $m$, however, quickly becomes prohibitive as they grow. Assuming that we analyze a $N$-qubit system and each measurement is stored in 9 bits (1 bit outcome, 8 bits Pauli rotation indicator), the memory requirements (excluding any intermediate data) will scale as $9 N T$. For example, a 6-qubit instance of 10,000 shots requires 540 KB, which already exceeds the typical L1 cache size. Thus, although the algorithm is easily parallelizable, cache locality is poor: CPU cores will spend most of their time idle waiting for data, rather than performing useful computation. This limits the achievable speed-up from concurrency.

Memory efficiency alone is therefore not sufficient to make the whole post-processing efficient and scalable. A very critical resource is \emph{data locality}, and our two online estimators, visualized in \Cref{fig:seq-moment-estimator}, represent the two opposites on the spectrum between data locality and memory efficiency. The reconstruction-free estimator improves memory usage and scales to larger systems sizes, but suffers from poor locality since each update iterates over an increasingly large portion of the measurements set. The reconstruction-based estimator, instead, stores $O(4^N)$ matrices, sacrificing memory to achieve data locality and constant per-shot update cost. 

Thus, there is no one-size-fits-all solution. In many use-cases, PT-moment estimation up to $m=2$ or $m=3$ suffices, and the reconstruction-free approach from \Cref{sec:streaming-no-reconstruction} remains practical even for moderate system sizes. In contrast, in scenarios where entanglement detection requires high-order PT moments, the reconstruction-based estimator in \Cref{sec:streaming-reconstructing-dm} would be the more appropriate choice, despite its exponential memory cost.

Note that our two estimators are on two opposite ends of the spectrum, which seemingly imposes a trade-off.  Typically, this refers, however, to balancing multiple objectives, such as memory usage and data locality, often expressed, for example, as a Pareto front showing the trade-offs among non-dominated solutions. In contrast, this section is actually addressing a binary choice between memory and data locality (“either–or”), rather than a continuum of possible compromises. It is not clear whether this problem can indeed be phrased in a way allowing to get (parts of) the best of both worlds for a more satisfactory solution balancing the two approaches.

\section{Results}\label{sec:results}

We evaluate both online estimators on a sequence of  Werner states~\cite{werner1989}. These are ideal test candidates, because they exist for any number of qubits, are mixed and feature a single parameter that controls the amount of entanglement present. As shown in Appendix~\ref{app:werner-states}, this parameter can be tuned to require progressively higher moment orders to detect entanglement. This allows us to generate test cases that are increasingly more challenging for our online estimators.

\begin{figure*}[t]
    \centering
    \includegraphics[width=0.9\linewidth]{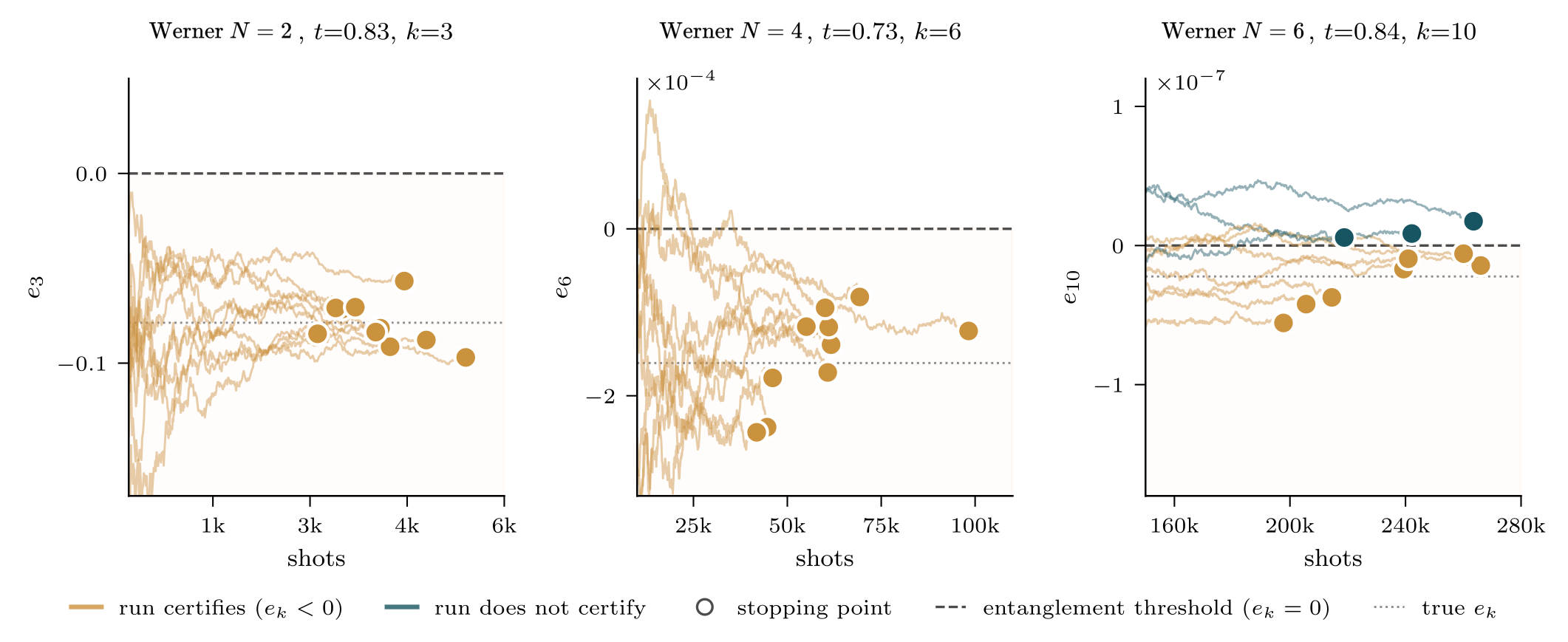}
    \caption{\textbf{Online entanglement detection across system sizes.} Each trace shows the elementary symmetric polynomial $e_k$ of the estimated PT spectrum over shots, for $10$ independent runs. Traces that cross $e_k = 0$ (dashed line) certify entanglement. Dots mark the stopping point of each run. Left: Werner $2$q, $k{=}3$. Middle: Werner $4$q, $k{=}6$. Right: Werner $6$q, $k{=}10$.}
    \label{fig:certification}
\end{figure*}

All experiments use $10$ independent datasets comprised of $10^6$ classical shadows (randomized single-qubit measurements) each, generated in silico from a classical description of the underlying target state. The online estimator with reconstruction (\Cref{sec:streaming-reconstructing-dm}) is implemented in Python~3.12 using sparse matrices from \texttt{scipy}~\cite{2020SciPy-NMeth}; the estimator without reconstruction (\Cref{sec:streaming-no-reconstruction}) uses \texttt{numpy}~\cite{numpy} with \texttt{numba}~\cite{numba} JIT compilation. 
We provide the source code and results in~\cite{GitHub}.
As a baseline, we compare against the batched shadow estimator from Ref.~\cite{batched-shadows}, implemented in the \texttt{RandomMeas.jl} package~\cite{elben2025randommeasjljuliapackagerandomized}. 
All experiments were run on an AMD EPYC 7402 server (48 cores, 1.5\,TB RAM).

The underlying Werner states can be written as $\rho_W(t) = (\mathbb{I} - tF)/(d^2 - dt)$, where $F$ denotes the SWAP operator. We execute benchmark studies for $N = 2, 4, 6, 8$ qubits and choose the mixing parameter $t$ such that the first violated moment order ranges from $k = 3$ to $k = 18$. The full set of benchmark instances is succinctly summarized in \Cref{tab:states-benchmarking} and we refer to \Cref{app:werner-states} for additional details. 

We declare convergence when the relative change between successive iterates stays below a tolerance threshold of $10^{-3}$ for at least $10$ consecutive steps, i.e.\ $|\hat{P}_{T+1}^{(m)} - \hat{P}_{T}^{(m)}| < 10^{-3} \max(|\hat{P}_{T+1}^{(m)}|, |\hat{P}_{T}^{(m)}|)$.
This type of relative-change stopping rule is widely used in iterative and online algorithms, such as fixed-point iterations~\cite[p.~59]{burden1997numerical}, iterative schemes in linear algebra (e.g., Jacobi~\cite[p.~439]{burden1997numerical} and Gauss-Seidel~\cite[p.~441]{burden1997numerical} methods) and Newton's method for root finding~\cite[p.~614]{burden1997numerical}.
It provides a practical handle on the accuracy vs runtime trade-off, even though it does not directly correspond to theoretical or sample complexity bounds. The specific choice of tolerance $10^{-3}$ and window length $10$ is heuristic and can be adjusted.

\subsection{Online entanglement detection}\label{sec:results-certification}

\Cref{fig:certification} shows the $e_k$ inequality traces produced by the reconstruction-based online estimator on Werner states comprised of $N=2$, $N=4$, and $N=6$ qubits. 
For a separable state, the elementary symmetric polynomial $e_k$ of the PT spectrum must be nonnegative; a negative value certifies entanglement. 
Across all three system sizes, the online traces cross the $e_k = 0$ threshold and converge to a stable negative value. 
For $N=2$ qubits ($k = 3$), convergence requires roughly $4\,000$ shots; for $N=4$ qubits ($k = 6$), roughly $60\,000$; at $6$ qubits ($k = 10$), roughly $250\,000$. 
The increase reflects the growing statistical difficulty of estimating higher-order moments of a state whose PT-moment violation shrinks with system size.

Not all runs detect entanglement for $N=6$ qubits: the violation is small enough that some runs converge before the estimate crosses zero. 
This is expected, as the true $e_{10}$ is of order $10^{-7}$, and the estimator variance at this moment order is comparable. 
Nevertheless, the majority of runs detect correctly, and increasing the shot budget would recover all runs.

\subsection{Comparison with batched shadows}\label{sec:results-comparison}

\Cref{fig:streaming-vs-batched} compares the online estimator against the batched shadow estimator of~\cite{batched-shadows} on three Werner state instances spanning different system sizes and moment orders. 
Each online trace shows the $e_k$ estimate as shots are processed; each batched point is a single offline estimate at a fixed shot budget. 
In all three cases the online estimator detects entanglement with substantially fewer shots than the batched baseline.
The $N{=}6$, $k{=}10$ instance from \Cref{fig:certification} is not included here as its true $e_{10}\approx 10^{-7}$ is smaller than the estimator variance of the batched method within our $10^6$-shot budget, leading to the batched estimator failing to detect entanglement at that instance and no meaningful comparison of sample efficiency is possible.

The underlying reason is that the online estimator evaluates the unbiased U-statistic over all $\binom{T}{m}$ combinations of shots, while the batched estimator partitions shots into disjoint blocks and averages within each block before forming the moment estimate. 
This discards cross-block correlations, reducing statistical efficiency.
The gap compounds as the moment order increases, since higher moments involve more cross-block terms that the batched estimator cannot access.

\begin{figure*}[t]
    \centering
    \includegraphics[width=0.9\linewidth]{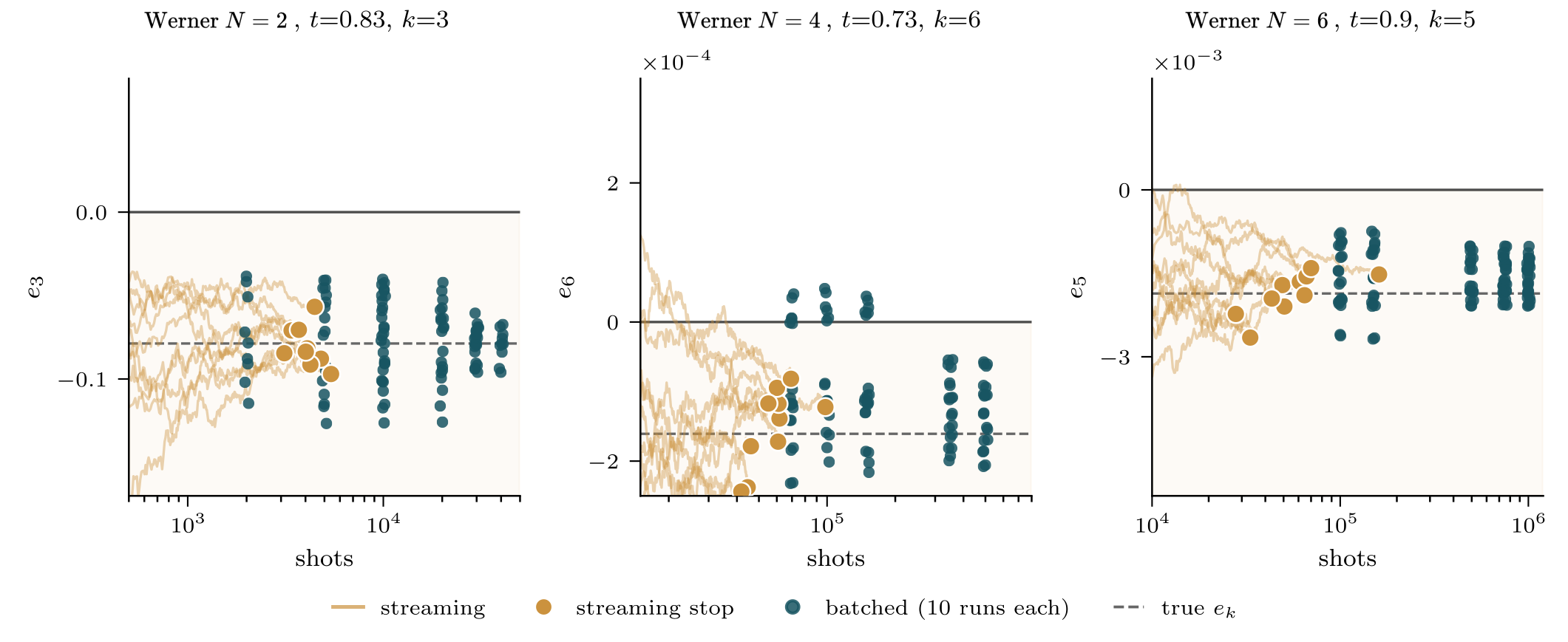}
    \caption{\textbf{Online vs.\ batched entanglement detection.}
    Orange traces show the online $e_k$ estimate over shots ($10$ runs); teal dots show batched estimates at fixed shot budgets ($10$ runs each). Dashed line: true $e_k$. Left: Werner $2$q, $t{=}0.83$, $k{=}3$. Middle: Werner $4$q, $t{=}0.73$, $k{=}6$. Right: Werner $6$q, $t{=}0.9$, $k{=}5$. The online estimator detects entanglement ($e_k < 0$) with fewer shots across all instances.}
    \label{fig:streaming-vs-batched}
\end{figure*}

\subsection{Runtime scaling}\label{sec:results-runtime}

\Cref{tab:runtimes} and \Cref{fig:runtime-scaling} report the mean wall-clock time for both online estimators across system sizes. 
The reconstruction-based estimator (\Cref{sec:streaming-reconstructing-dm}) processes all moment orders in a single pass with constant per-shot cost, but its $O(m \cdot 4^N)$ memory footprint limits it to moderate qubit numbers $N \lesssim 8$; for $N=6$ qubits the run takes roughly $10$ minutes.

The estimator without reconstruction (\Cref{sec:streaming-no-reconstruction}) avoids exponential memory and scales to $10$ qubits for the purity ($m = 2$, $13$~minutes), but its $\Theta(T^{m-1})$ per-shot cost makes $m = 3$ at $10$ qubits already expensive ($\sim\!36$~hours)
The jump from $m = 2$ to $m = 3$ for $N=8$ and $N=10$ qubits, which is visible in \Cref{fig:runtime-scaling}, reflects the polynomial-in-$T$ cost of the combinatorial update, which dominates wall-clock time once the dataset is large enough to exceed the fastest caches.

The two estimators thus occupy complementary regimes, as anticipated by the analysis in \Cref{sec:streaming-reconstructing-dm}: the reconstruction-based approach is the method of choice when high moment orders are needed and the system fits in memory, while the reconstruction-free approach is preferable for large systems at low moment order.

In this work, we refrain from extensive runtime comparisons due to (1) unequal technical setups based on Python vs. Julia, and (2), the fact that a runtime comparison would be only meaningful for real world set-ups, where the data acquisition is done in real-time to account for sampling time as well as device overheads. Note that this will significantly differ based on the platform and exact device, thus, extensive and reliable empirical runtime benchmarks are a fully fledged research question of their own and out of scope for this first proposal.

\begin{table}[]
    \centering
    
    \begin{tabular}{lrrrrr}
\toprule
Method & 2q & 4q & 6q & 8q & 10q \\
\midrule
With reconstruction (any $k$) & 2.3\,s & 1\,min & 10\,min & 23\,h & --- \\
\midrule
Without recon., $m{=}2$ & 2\,s & 0.5\,s & --- & 16\,min & 13\,min \\
Without recon., $m{=}3$ & 1\,s & 12\,s & --- & 3.3\,h & 36\,h \\
\bottomrule
\end{tabular}
\caption{Mean runtime of the two online estimators across system sizes. Dashes indicate configurations not run.}
\label{tab:runtimes}
\end{table}

\begin{figure}[t]
    \centering
    \includegraphics[width=\linewidth]{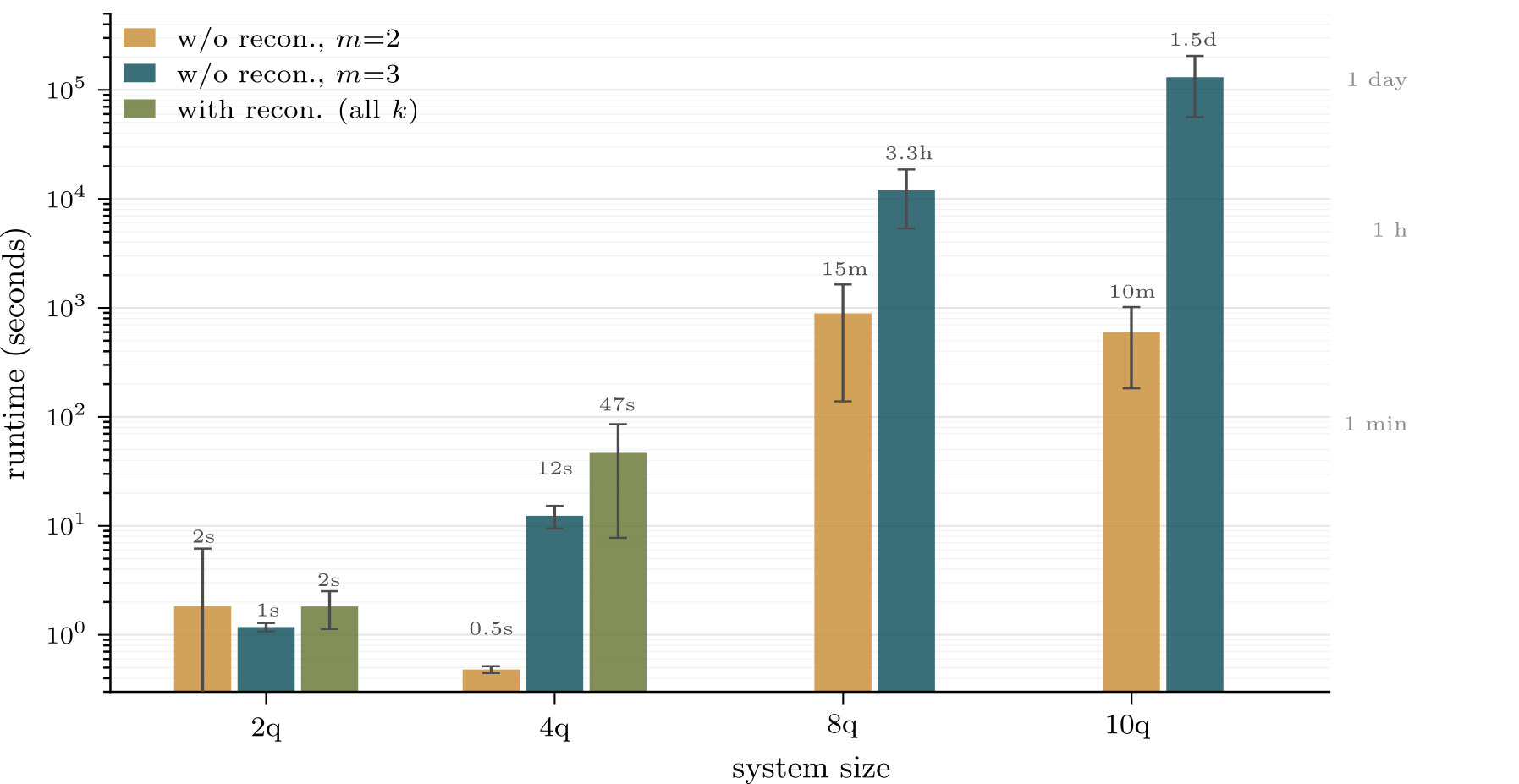}
    \caption{\textbf{Runtime scaling of both online estimators.}
    Bars show mean wall-clock time (log scale) across $10$ runs; error bars indicate one standard deviation. The without-reconstruction estimator (orange, teal) scales to $10$ qubits at $m{=}2$ but becomes expensive at $m{=}3$. The reconstruction-based estimator (olive) handles all moment orders but is limited to smaller systems by its $O(4^N)$ memory footprint. Reference lines mark $1$~minute, $1$~hour, and $1$~day.}
    \label{fig:runtime-scaling}
\end{figure}
\section{Conclusion}\label{sec:discussion-conclusion}
The efficiency of hybrid quantum-classical algorithms is limited by the efficiency of the individual components, and, in current NISQ experiments~\cite{Preskill2018quantumcomputingin}, significant resources are wasted as classical processors sit idle between shots.
This work takes a concrete step toward closing this gap for a concrete and practically relevant task: entanglement detection in highly mixed states.
We have introduced an online framework for classical shadows in which PT-moment estimators~\cite{moments-ppt,moments-ppt-2,10.21468/SciPostPhys.12.3.106} are updated incrementally with each new measurement, turning reset latency into time for useful classical computation.

Two estimators emerge from this framework, each occupying one end of the spectrum between memory footprint and per-shot computational cost: one processes the raw measurement record with memory linear in the number of shots, suitable for large systems at low moment order; the other maintains a fixed set of $m$ accumulated matrices with constant per-shot cost, thus, suitable for high moment orders at small system size.

The empirical results confirm the practical value of the online approach. On Werner states across $2$ to $6$ qubits, the online estimator detects entanglement with substantially fewer shots than the batched shadow baseline, because the U-statistic exploits all $\binom{T}{m}$ combinations of snapshots rather than discarding cross-block correlations. 
The two estimators complement each other as the reconstruction-based variant handles moment orders up to $k = 18$ at $8$ qubits, while the reconstruction-free variant scales to $10$ qubits for low moment orders. 
Both are unbiased by construction and can run concurrently with the experiment.

The current limits are well-characterized. High moment orders at large system sizes remain out of reach: the reconstruction-based estimator is restricted by its $O(m \cdot 4^N)$ memory footprint, which limits it to $N \le 8$ qubits on current hardware; the reconstruction-free estimator is limited by its $\Theta(T^{m-1})$ per-shot cost, which becomes dominant beyond moment order three at ten qubits. These are not artifacts of the implementation but consequences of the combinatorial structure of nonlinear estimation from single-copy measurements. Indeed, rigorous no-go theorems tell us that nonlinear feature estimation from single-copy measurements must require exponential resources (in the number of qubits) in general~\cite{9719827,10756089}.

Two directions are particularly promising. The first is extending the online framework to shadow protocols beyond local Clifford measurements, as well as multi-copy measurement schemes that could bypass the aforementioned information-theoretic no-go theorems and reduce the combinatorial cost intrinsic to nonlinear estimation~\cite{PhysRevLett.126.190505,9719827,10.1126/science.abn7293,noller2025infinite}. 
The second direction is demonstrating the online advantage in a real experimental loop on NISQ hardware. Current superconducting devices operate at measurement rates of roughly $1$-$10$\,kHz~\cite{supremacy}, leaving milliseconds of idle time between shots, more than sufficient for a low-order online update. 
Trapped-ion devices are slower still, making the online opportunity larger. 
A hardware demonstration would validate the central premise of this work under realistic noise and latency conditions.

More broadly, the offline assumption that separates data acquisition from post-processing is not fundamental to hybrid quantum-classical computation. 
Instead, it is merely a design choice, and one that systematically underutilizes classical resources. 
Classical shadows, with their independent single-shot structure, are uniquely positioned to challenge this assumption. 
The online paradigm is not merely a computational optimization but the natural architecture for any protocol in which classical processing can keep pace with quantum data generation.

\acknowledgments
SH is a recipient of the Andreas Dieberger-Peter Skalicky scholarship and acknowledges financial support by the netidee funding campaign (grant 7728). This work has been partially funded through the Themis project (Trustworthy and Sustainable Code Offloading), Austrian Science Fund (FWF): Grant-DOI 10.55776/PAT1668223, the Standalone Project Transprecise Edge Computing (Triton), Austrian Science Fund (FWF): P 36870-N, and by the Flagship Project HPQC (High Performance Integrated Quantum Computing) \# 897481 Austrian Research Promotion Agency (FFG). Research at JKU has in addition been supported by the European Research Council (ERC) via the Starting grant q-shadows (101117138) and from the Austrian Science Fund (FWF) via the SFB BeyondC (10.55776/FG7).

\bibliography{entanglement-verification}

\begin{thebibliography}{32}%
\makeatletter
\providecommand \@ifxundefined [1]{%
 \@ifx{#1\undefined}
}%
\providecommand \@ifnum [1]{%
 \ifnum #1\expandafter \@firstoftwo
 \else \expandafter \@secondoftwo
 \fi
}%
\providecommand \@ifx [1]{%
 \ifx #1\expandafter \@firstoftwo
 \else \expandafter \@secondoftwo
 \fi
}%
\providecommand \natexlab [1]{#1}%
\providecommand \enquote  [1]{``#1''}%
\providecommand \bibnamefont  [1]{#1}%
\providecommand \bibfnamefont [1]{#1}%
\providecommand \citenamefont [1]{#1}%
\providecommand \href@noop [0]{\@secondoftwo}%
\providecommand \href [0]{\begingroup \@sanitize@url \@href}%
\providecommand \@href[1]{\@@startlink{#1}\@@href}%
\providecommand \@@href[1]{\endgroup#1\@@endlink}%
\providecommand \@sanitize@url [0]{\catcode `\\12\catcode `\$12\catcode `\&12\catcode `\#12\catcode `\^12\catcode `\_12\catcode `\%12\relax}%
\providecommand \@@startlink[1]{}%
\providecommand \@@endlink[0]{}%
\providecommand \url  [0]{\begingroup\@sanitize@url \@url }%
\providecommand \@url [1]{\endgroup\@href {#1}{\urlprefix }}%
\providecommand \urlprefix  [0]{URL }%
\providecommand \Eprint [0]{\href }%
\providecommand \doibase [0]{https://doi.org/}%
\providecommand \selectlanguage [0]{\@gobble}%
\providecommand \bibinfo  [0]{\@secondoftwo}%
\providecommand \bibfield  [0]{\@secondoftwo}%
\providecommand \translation [1]{[#1]}%
\providecommand \BibitemOpen [0]{}%
\providecommand \bibitemStop [0]{}%
\providecommand \bibitemNoStop [0]{.\EOS\space}%
\providecommand \EOS [0]{\spacefactor3000\relax}%
\providecommand \BibitemShut  [1]{\csname bibitem#1\endcsname}%
\let\auto@bib@innerbib\@empty
\bibitem [{\citenamefont {Preskill}(2018)}]{Preskill2018quantumcomputingin}%
  \BibitemOpen
  \bibfield  {author} {\bibinfo {author} {\bibfnamefont {J.}~\bibnamefont {Preskill}},\ }\bibfield  {title} {\bibinfo {title} {Quantum {C}omputing in the {NISQ} era and beyond},\ }\href {https://doi.org/10.22331/q-2018-08-06-79} {\bibfield  {journal} {\bibinfo  {journal} {{Quantum}}\ }\textbf {\bibinfo {volume} {2}},\ \bibinfo {pages} {79} (\bibinfo {year} {2018})}\BibitemShut {NoStop}%
\bibitem [{\citenamefont {Hoefler}\ \emph {et~al.}(2023)\citenamefont {Hoefler}, \citenamefont {H{\"a}ner},\ and\ \citenamefont {Troyer}}]{hoefler2023disentangling}%
  \BibitemOpen
  \bibfield  {author} {\bibinfo {author} {\bibfnamefont {T.}~\bibnamefont {Hoefler}}, \bibinfo {author} {\bibfnamefont {T.}~\bibnamefont {H{\"a}ner}},\ and\ \bibinfo {author} {\bibfnamefont {M.}~\bibnamefont {Troyer}},\ }\bibfield  {title} {\bibinfo {title} {Disentangling hype from practicality: On realistically achieving quantum advantage},\ }\href@noop {} {\bibfield  {journal} {\bibinfo  {journal} {Communications of the ACM}\ }\textbf {\bibinfo {volume} {66}},\ \bibinfo {pages} {82} (\bibinfo {year} {2023})}\BibitemShut {NoStop}%
\bibitem [{\citenamefont {Arute}\ \emph {et~al.}(2019)\citenamefont {Arute}, \citenamefont {Arya}, \citenamefont {Babbush} \emph {et~al.}}]{supremacy}%
  \BibitemOpen
  \bibfield  {author} {\bibinfo {author} {\bibfnamefont {F.}~\bibnamefont {Arute}}, \bibinfo {author} {\bibfnamefont {K.}~\bibnamefont {Arya}}, \bibinfo {author} {\bibfnamefont {R.}~\bibnamefont {Babbush}}, \emph {et~al.},\ }\bibfield  {title} {\bibinfo {title} {Quantum supremacy using a programmable superconducting processor},\ }\href {https://doi.org/10.1038/s41586-019-1666-5} {\bibfield  {journal} {\bibinfo  {journal} {Nature}\ }\textbf {\bibinfo {volume} {574}},\ \bibinfo {pages} {505} (\bibinfo {year} {2019})}\BibitemShut {NoStop}%
\bibitem [{\citenamefont {Manasse}\ \emph {et~al.}(1988)\citenamefont {Manasse}, \citenamefont {McGeoch},\ and\ \citenamefont {Sleator}}]{manasse1988competitive}%
  \BibitemOpen
  \bibfield  {author} {\bibinfo {author} {\bibfnamefont {M.}~\bibnamefont {Manasse}}, \bibinfo {author} {\bibfnamefont {L.}~\bibnamefont {McGeoch}},\ and\ \bibinfo {author} {\bibfnamefont {D.}~\bibnamefont {Sleator}},\ }\bibfield  {title} {\bibinfo {title} {Competitive algorithms for on-line problems},\ }in\ \href@noop {} {\emph {\bibinfo {booktitle} {Proceedings of the twentieth annual ACM symposium on Theory of computing}}}\ (\bibinfo {year} {1988})\ pp.\ \bibinfo {pages} {322--333}\BibitemShut {NoStop}%
\bibitem [{\citenamefont {Elben}\ \emph {et~al.}(2023)\citenamefont {Elben}, \citenamefont {Flammia}, \citenamefont {Huang}, \citenamefont {Kueng}, \citenamefont {Preskill}, \citenamefont {Vermersch},\ and\ \citenamefont {Zoller}}]{elben2023randomized}%
  \BibitemOpen
  \bibfield  {author} {\bibinfo {author} {\bibfnamefont {A.}~\bibnamefont {Elben}}, \bibinfo {author} {\bibfnamefont {S.~T.}\ \bibnamefont {Flammia}}, \bibinfo {author} {\bibfnamefont {H.-Y.}\ \bibnamefont {Huang}}, \bibinfo {author} {\bibfnamefont {R.}~\bibnamefont {Kueng}}, \bibinfo {author} {\bibfnamefont {J.}~\bibnamefont {Preskill}}, \bibinfo {author} {\bibfnamefont {B.}~\bibnamefont {Vermersch}},\ and\ \bibinfo {author} {\bibfnamefont {P.}~\bibnamefont {Zoller}},\ }\bibfield  {title} {\bibinfo {title} {The randomized measurement toolbox},\ }\href@noop {} {\bibfield  {journal} {\bibinfo  {journal} {Nature Reviews Physics}\ }\textbf {\bibinfo {volume} {5}},\ \bibinfo {pages} {9} (\bibinfo {year} {2023})}\BibitemShut {NoStop}%
\bibitem [{\citenamefont {Huang}\ \emph {et~al.}(2020)\citenamefont {Huang}, \citenamefont {Kueng},\ and\ \citenamefont {Preskill}}]{huang-predicting}%
  \BibitemOpen
  \bibfield  {author} {\bibinfo {author} {\bibfnamefont {H.}~\bibnamefont {Huang}}, \bibinfo {author} {\bibfnamefont {R.}~\bibnamefont {Kueng}},\ and\ \bibinfo {author} {\bibfnamefont {J.}~\bibnamefont {Preskill}},\ }\bibfield  {title} {\bibinfo {title} {Predicting many properties of a quantum system from very few measurements},\ }\href {https://doi.org/10.1038/s41567-020-0932-7} {\bibfield  {journal} {\bibinfo  {journal} {Nat. Phys.}\ }\textbf {\bibinfo {volume} {16}},\ \bibinfo {pages} {1050} (\bibinfo {year} {2020})}\BibitemShut {NoStop}%
\bibitem [{\citenamefont {Paini}\ and\ \citenamefont {Kalev}(2019{\natexlab{a}})}]{paini2019approximatedescriptionquantumstates}%
  \BibitemOpen
  \bibfield  {author} {\bibinfo {author} {\bibfnamefont {M.}~\bibnamefont {Paini}}\ and\ \bibinfo {author} {\bibfnamefont {A.}~\bibnamefont {Kalev}},\ }\href {https://arxiv.org/abs/1910.10543} {\bibinfo {title} {An approximate description of quantum states}} (\bibinfo {year} {2019}{\natexlab{a}}),\ \Eprint {https://arxiv.org/abs/1910.10543} {arXiv:1910.10543 [quant-ph]} \BibitemShut {NoStop}%
\bibitem [{\citenamefont {Morris}\ and\ \citenamefont {Dakić}(2020)}]{morris2020selectivequantumstatetomography}%
  \BibitemOpen
  \bibfield  {author} {\bibinfo {author} {\bibfnamefont {J.}~\bibnamefont {Morris}}\ and\ \bibinfo {author} {\bibfnamefont {B.}~\bibnamefont {Dakić}},\ }\href {https://arxiv.org/abs/1909.05880} {\bibinfo {title} {Selective quantum state tomography}} (\bibinfo {year} {2020}),\ \Eprint {https://arxiv.org/abs/1909.05880} {arXiv:1909.05880 [quant-ph]} \BibitemShut {NoStop}%
\bibitem [{\citenamefont {Peres}(1996)}]{Peres1996}%
  \BibitemOpen
  \bibfield  {author} {\bibinfo {author} {\bibfnamefont {A.}~\bibnamefont {Peres}},\ }\bibfield  {title} {\bibinfo {title} {Separability criterion for density matrices},\ }\href {https://doi.org/10.1103/physrevlett.77.1413} {\bibfield  {journal} {\bibinfo  {journal} {Physical Review Letters}\ }\textbf {\bibinfo {volume} {77}},\ \bibinfo {pages} {1413–1415} (\bibinfo {year} {1996})}\BibitemShut {NoStop}%
\bibitem [{\citenamefont {Horodecki}\ \emph {et~al.}(1996)\citenamefont {Horodecki}, \citenamefont {Horodecki},\ and\ \citenamefont {Horodecki}}]{Horodecki1996}%
  \BibitemOpen
  \bibfield  {author} {\bibinfo {author} {\bibfnamefont {M.}~\bibnamefont {Horodecki}}, \bibinfo {author} {\bibfnamefont {P.}~\bibnamefont {Horodecki}},\ and\ \bibinfo {author} {\bibfnamefont {R.}~\bibnamefont {Horodecki}},\ }\bibfield  {title} {\bibinfo {title} {Separability of mixed states: necessary and sufficient conditions},\ }\href {https://doi.org/10.1016/s0375-9601(96)00706-2} {\bibfield  {journal} {\bibinfo  {journal} {Physics Letters A}\ }\textbf {\bibinfo {volume} {223}},\ \bibinfo {pages} {1–8} (\bibinfo {year} {1996})}\BibitemShut {NoStop}%
\bibitem [{\citenamefont {Elben}\ \emph {et~al.}(2020)\citenamefont {Elben}, \citenamefont {Kueng}, \citenamefont {Huang}, \citenamefont {van Bijnen}, \citenamefont {Kokail}, \citenamefont {Dalmonte}, \citenamefont {Calabrese}, \citenamefont {Kraus}, \citenamefont {Preskill}, \citenamefont {Zoller},\ and\ \citenamefont {Vermersch}}]{moments-ppt}%
  \BibitemOpen
  \bibfield  {author} {\bibinfo {author} {\bibfnamefont {A.}~\bibnamefont {Elben}}, \bibinfo {author} {\bibfnamefont {R.}~\bibnamefont {Kueng}}, \bibinfo {author} {\bibfnamefont {H.-Y.~R.}\ \bibnamefont {Huang}}, \bibinfo {author} {\bibfnamefont {R.}~\bibnamefont {van Bijnen}}, \bibinfo {author} {\bibfnamefont {C.}~\bibnamefont {Kokail}}, \bibinfo {author} {\bibfnamefont {M.}~\bibnamefont {Dalmonte}}, \bibinfo {author} {\bibfnamefont {P.}~\bibnamefont {Calabrese}}, \bibinfo {author} {\bibfnamefont {B.}~\bibnamefont {Kraus}}, \bibinfo {author} {\bibfnamefont {J.}~\bibnamefont {Preskill}}, \bibinfo {author} {\bibfnamefont {P.}~\bibnamefont {Zoller}},\ and\ \bibinfo {author} {\bibfnamefont {B.}~\bibnamefont {Vermersch}},\ }\bibfield  {title} {\bibinfo {title} {Mixed-state entanglement from local randomized measurements},\ }\href {https://doi.org/10.1103/PhysRevLett.125.200501} {\bibfield  {journal} {\bibinfo  {journal} {Phys. Rev. Lett.}\ }\textbf {\bibinfo {volume} {125}},\ \bibinfo {pages} {200501} (\bibinfo
  {year} {2020})}\BibitemShut {NoStop}%
\bibitem [{\citenamefont {Neven}\ \emph {et~al.}(2021)\citenamefont {Neven}, \citenamefont {Carrasco}, \citenamefont {Vitale},\ and\ \citenamefont {Others}}]{moments-ppt-2}%
  \BibitemOpen
  \bibfield  {author} {\bibinfo {author} {\bibfnamefont {A.}~\bibnamefont {Neven}}, \bibinfo {author} {\bibfnamefont {J.}~\bibnamefont {Carrasco}}, \bibinfo {author} {\bibfnamefont {V.}~\bibnamefont {Vitale}},\ and\ \bibinfo {author} {\bibnamefont {Others}},\ }\bibfield  {title} {\bibinfo {title} {Symmetry-resolved entanglement detection using partial transpose moments},\ }\href {10.1038/s41534-021-00487-y} {\bibfield  {journal} {\bibinfo  {journal} {npj Quantum Inf}\ }\textbf {\bibinfo {volume} {7}},\ \bibinfo {pages} {1} (\bibinfo {year} {2021})}\BibitemShut {NoStop}%
\bibitem [{\citenamefont {Vitale}\ \emph {et~al.}(2022)\citenamefont {Vitale}, \citenamefont {Elben}, \citenamefont {Kueng}, \citenamefont {Neven}, \citenamefont {Carrasco}, \citenamefont {Kraus}, \citenamefont {Zoller}, \citenamefont {Calabrese}, \citenamefont {Vermersch},\ and\ \citenamefont {Dalmonte}}]{10.21468/SciPostPhys.12.3.106}%
  \BibitemOpen
  \bibfield  {author} {\bibinfo {author} {\bibfnamefont {V.}~\bibnamefont {Vitale}}, \bibinfo {author} {\bibfnamefont {A.}~\bibnamefont {Elben}}, \bibinfo {author} {\bibfnamefont {R.}~\bibnamefont {Kueng}}, \bibinfo {author} {\bibfnamefont {A.}~\bibnamefont {Neven}}, \bibinfo {author} {\bibfnamefont {J.}~\bibnamefont {Carrasco}}, \bibinfo {author} {\bibfnamefont {B.}~\bibnamefont {Kraus}}, \bibinfo {author} {\bibfnamefont {P.}~\bibnamefont {Zoller}}, \bibinfo {author} {\bibfnamefont {P.}~\bibnamefont {Calabrese}}, \bibinfo {author} {\bibfnamefont {B.}~\bibnamefont {Vermersch}},\ and\ \bibinfo {author} {\bibfnamefont {M.}~\bibnamefont {Dalmonte}},\ }\bibfield  {title} {\bibinfo {title} {{Symmetry-resolved dynamical purification in synthetic quantum matter}},\ }\href {https://doi.org/10.21468/SciPostPhys.12.3.106} {\bibfield  {journal} {\bibinfo  {journal} {SciPost Phys.}\ }\textbf {\bibinfo {volume} {12}},\ \bibinfo {pages} {106} (\bibinfo {year} {2022})}\BibitemShut {NoStop}%
\bibitem [{\citenamefont {Rath}\ \emph {et~al.}(2023)\citenamefont {Rath}, \citenamefont {Vitale}, \citenamefont {Murciano}, \citenamefont {Votto}, \citenamefont {Dubail}, \citenamefont {Kueng}, \citenamefont {Branciard}, \citenamefont {Calabrese},\ and\ \citenamefont {Vermersch}}]{batched-shadows}%
  \BibitemOpen
  \bibfield  {author} {\bibinfo {author} {\bibfnamefont {A.}~\bibnamefont {Rath}}, \bibinfo {author} {\bibfnamefont {V.}~\bibnamefont {Vitale}}, \bibinfo {author} {\bibfnamefont {S.}~\bibnamefont {Murciano}}, \bibinfo {author} {\bibfnamefont {M.}~\bibnamefont {Votto}}, \bibinfo {author} {\bibfnamefont {J.}~\bibnamefont {Dubail}}, \bibinfo {author} {\bibfnamefont {R.}~\bibnamefont {Kueng}}, \bibinfo {author} {\bibfnamefont {C.}~\bibnamefont {Branciard}}, \bibinfo {author} {\bibfnamefont {P.}~\bibnamefont {Calabrese}},\ and\ \bibinfo {author} {\bibfnamefont {B.}~\bibnamefont {Vermersch}},\ }\bibfield  {title} {\bibinfo {title} {Entanglement barrier and its symmetry resolution: Theory and experimental observation},\ }\href {https://doi.org/10.1103/PRXQuantum.4.010318} {\bibfield  {journal} {\bibinfo  {journal} {PRX Quantum}\ }\textbf {\bibinfo {volume} {4}},\ \bibinfo {pages} {010318} (\bibinfo {year} {2023})}\BibitemShut {NoStop}%
\bibitem [{\citenamefont {Paini}\ and\ \citenamefont {Kalev}(2019{\natexlab{b}})}]{paini2019approximate}%
  \BibitemOpen
  \bibfield  {author} {\bibinfo {author} {\bibfnamefont {M.}~\bibnamefont {Paini}}\ and\ \bibinfo {author} {\bibfnamefont {A.}~\bibnamefont {Kalev}},\ }\bibfield  {title} {\bibinfo {title} {An approximate description of quantum states},\ }\href@noop {} {\bibfield  {journal} {\bibinfo  {journal} {arXiv preprint arXiv:1910.10543}\ } (\bibinfo {year} {2019}{\natexlab{b}})}\BibitemShut {NoStop}%
\bibitem [{\citenamefont {Morris}\ and\ \citenamefont {Daki{\'c}}(2019)}]{morris2019selective}%
  \BibitemOpen
  \bibfield  {author} {\bibinfo {author} {\bibfnamefont {J.}~\bibnamefont {Morris}}\ and\ \bibinfo {author} {\bibfnamefont {B.}~\bibnamefont {Daki{\'c}}},\ }\bibfield  {title} {\bibinfo {title} {Selective quantum state tomography},\ }\href@noop {} {\bibfield  {journal} {\bibinfo  {journal} {arXiv preprint arXiv:1909.05880}\ } (\bibinfo {year} {2019})}\BibitemShut {NoStop}%
\bibitem [{\citenamefont {Chen}\ \emph {et~al.}(2022)\citenamefont {Chen}, \citenamefont {Cotler}, \citenamefont {Huang},\ and\ \citenamefont {Li}}]{9719827}%
  \BibitemOpen
  \bibfield  {author} {\bibinfo {author} {\bibfnamefont {S.}~\bibnamefont {Chen}}, \bibinfo {author} {\bibfnamefont {J.}~\bibnamefont {Cotler}}, \bibinfo {author} {\bibfnamefont {H.-Y.}\ \bibnamefont {Huang}},\ and\ \bibinfo {author} {\bibfnamefont {J.}~\bibnamefont {Li}},\ }\bibfield  {title} {\bibinfo {title} {Exponential separations between learning with and without quantum memory},\ }in\ \href@noop {} {\emph {\bibinfo {booktitle} {2021 IEEE 62nd Annual Symposium on Foundations of Computer Science (FOCS)}}}\ (\bibinfo  {publisher} {IEEE},\ \bibinfo {address} {Denver, CO, USA},\ \bibinfo {year} {2022})\ pp.\ \bibinfo {pages} {574--585}\BibitemShut {NoStop}%
\bibitem [{\citenamefont {Chen}\ \emph {et~al.}(2024)\citenamefont {Chen}, \citenamefont {Gong},\ and\ \citenamefont {Ye}}]{10756089}%
  \BibitemOpen
  \bibfield  {author} {\bibinfo {author} {\bibfnamefont {S.}~\bibnamefont {Chen}}, \bibinfo {author} {\bibfnamefont {W.}~\bibnamefont {Gong}},\ and\ \bibinfo {author} {\bibfnamefont {Q.}~\bibnamefont {Ye}},\ }\bibfield  {title} {\bibinfo {title} {Optimal tradeoffs for estimating pauli observables},\ }in\ \href {https://doi.org/10.1109/FOCS61266.2024.00072} {\emph {\bibinfo {booktitle} {2024 IEEE 65th Annual Symposium on Foundations of Computer Science (FOCS)}}}\ (\bibinfo  {publisher} {IEEE},\ \bibinfo {address} {Chicago, IL, USA},\ \bibinfo {year} {2024})\ pp.\ \bibinfo {pages} {1086--1105}\BibitemShut {NoStop}%
\bibitem [{\citenamefont {Sugiyama}\ \emph {et~al.}(2013)\citenamefont {Sugiyama}, \citenamefont {Turner},\ and\ \citenamefont {Murao}}]{PhysRevLett.111.160406}%
  \BibitemOpen
  \bibfield  {author} {\bibinfo {author} {\bibfnamefont {T.}~\bibnamefont {Sugiyama}}, \bibinfo {author} {\bibfnamefont {P.~S.}\ \bibnamefont {Turner}},\ and\ \bibinfo {author} {\bibfnamefont {M.}~\bibnamefont {Murao}},\ }\bibfield  {title} {\bibinfo {title} {Precision-guaranteed quantum tomography},\ }\href {https://doi.org/10.1103/PhysRevLett.111.160406} {\bibfield  {journal} {\bibinfo  {journal} {Phys. Rev. Lett.}\ }\textbf {\bibinfo {volume} {111}},\ \bibinfo {pages} {160406} (\bibinfo {year} {2013})}\BibitemShut {NoStop}%
\bibitem [{\citenamefont {Guţă}\ \emph {et~al.}(2020)\citenamefont {Guţă}, \citenamefont {Kahn}, \citenamefont {Kueng},\ and\ \citenamefont {Tropp}}]{Guta_2020}%
  \BibitemOpen
  \bibfield  {author} {\bibinfo {author} {\bibfnamefont {M.}~\bibnamefont {Guţă}}, \bibinfo {author} {\bibfnamefont {J.}~\bibnamefont {Kahn}}, \bibinfo {author} {\bibfnamefont {R.}~\bibnamefont {Kueng}},\ and\ \bibinfo {author} {\bibfnamefont {J.~A.}\ \bibnamefont {Tropp}},\ }\bibfield  {title} {\bibinfo {title} {Fast state tomography with optimal error bounds},\ }\href {https://doi.org/10.1088/1751-8121/ab8111} {\bibfield  {journal} {\bibinfo  {journal} {Journal of Physics A: Mathematical and Theoretical}\ }\textbf {\bibinfo {volume} {53}},\ \bibinfo {pages} {204001} (\bibinfo {year} {2020})}\BibitemShut {NoStop}%
\bibitem [{\citenamefont {Gunyh{\'o}}\ \emph {et~al.}(2024)\citenamefont {Gunyh{\'o}}, \citenamefont {Kundu}, \citenamefont {Ma}, \citenamefont {Liu}, \citenamefont {Niemel{\"a}}, \citenamefont {Catto}, \citenamefont {Vadimov}, \citenamefont {Vesterinen}, \citenamefont {Singh}, \citenamefont {Chen} \emph {et~al.}}]{gunyho2024single}%
  \BibitemOpen
  \bibfield  {author} {\bibinfo {author} {\bibfnamefont {A.~M.}\ \bibnamefont {Gunyh{\'o}}}, \bibinfo {author} {\bibfnamefont {S.}~\bibnamefont {Kundu}}, \bibinfo {author} {\bibfnamefont {J.}~\bibnamefont {Ma}}, \bibinfo {author} {\bibfnamefont {W.}~\bibnamefont {Liu}}, \bibinfo {author} {\bibfnamefont {S.}~\bibnamefont {Niemel{\"a}}}, \bibinfo {author} {\bibfnamefont {G.}~\bibnamefont {Catto}}, \bibinfo {author} {\bibfnamefont {V.}~\bibnamefont {Vadimov}}, \bibinfo {author} {\bibfnamefont {V.}~\bibnamefont {Vesterinen}}, \bibinfo {author} {\bibfnamefont {P.}~\bibnamefont {Singh}}, \bibinfo {author} {\bibfnamefont {Q.}~\bibnamefont {Chen}}, \emph {et~al.},\ }\bibfield  {title} {\bibinfo {title} {Single-shot readout of a superconducting qubit using a thermal detector},\ }\href@noop {} {\bibfield  {journal} {\bibinfo  {journal} {Nature Electronics}\ }\textbf {\bibinfo {volume} {7}},\ \bibinfo {pages} {288} (\bibinfo {year} {2024})}\BibitemShut {NoStop}%
\bibitem [{\citenamefont {Walter}\ \emph {et~al.}(2017)\citenamefont {Walter}, \citenamefont {Kurpiers}, \citenamefont {Gasparinetti}, \citenamefont {Magnard}, \citenamefont {Poto\ifmmode~\check{c}\else \v{c}\fi{}nik}, \citenamefont {Salath\'e}, \citenamefont {Pechal}, \citenamefont {Mondal}, \citenamefont {Oppliger}, \citenamefont {Eichler},\ and\ \citenamefont {Wallraff}}]{PhysRevApplied.7.054020}%
  \BibitemOpen
  \bibfield  {author} {\bibinfo {author} {\bibfnamefont {T.}~\bibnamefont {Walter}}, \bibinfo {author} {\bibfnamefont {P.}~\bibnamefont {Kurpiers}}, \bibinfo {author} {\bibfnamefont {S.}~\bibnamefont {Gasparinetti}}, \bibinfo {author} {\bibfnamefont {P.}~\bibnamefont {Magnard}}, \bibinfo {author} {\bibfnamefont {A.}~\bibnamefont {Poto\ifmmode~\check{c}\else \v{c}\fi{}nik}}, \bibinfo {author} {\bibfnamefont {Y.}~\bibnamefont {Salath\'e}}, \bibinfo {author} {\bibfnamefont {M.}~\bibnamefont {Pechal}}, \bibinfo {author} {\bibfnamefont {M.}~\bibnamefont {Mondal}}, \bibinfo {author} {\bibfnamefont {M.}~\bibnamefont {Oppliger}}, \bibinfo {author} {\bibfnamefont {C.}~\bibnamefont {Eichler}},\ and\ \bibinfo {author} {\bibfnamefont {A.}~\bibnamefont {Wallraff}},\ }\bibfield  {title} {\bibinfo {title} {Rapid high-fidelity single-shot dispersive readout of superconducting qubits},\ }\href {https://doi.org/10.1103/PhysRevApplied.7.054020} {\bibfield  {journal} {\bibinfo  {journal} {Phys. Rev. Appl.}\ }\textbf {\bibinfo
  {volume} {7}},\ \bibinfo {pages} {054020} (\bibinfo {year} {2017})}\BibitemShut {NoStop}%
\bibitem [{\citenamefont {Werner}(1989)}]{werner1989}%
  \BibitemOpen
  \bibfield  {author} {\bibinfo {author} {\bibfnamefont {R.~F.}\ \bibnamefont {Werner}},\ }\bibfield  {title} {\bibinfo {title} {Quantum states with einstein-podolsky-rosen correlations admitting a hidden-variable model},\ }\href {https://doi.org/10.1103/PhysRevA.40.4277} {\bibfield  {journal} {\bibinfo  {journal} {Phys. Rev. A}\ }\textbf {\bibinfo {volume} {40}},\ \bibinfo {pages} {4277} (\bibinfo {year} {1989})}\BibitemShut {NoStop}%
\bibitem [{\citenamefont {Virtanen}\ \emph {et~al.}(2020)\citenamefont {Virtanen}, \citenamefont {Gommers}, \citenamefont {Oliphant}, \citenamefont {Haberland}, \citenamefont {Reddy}, \citenamefont {Cournapeau}, \citenamefont {Burovski}, \citenamefont {Peterson}, \citenamefont {Weckesser}, \citenamefont {Bright}, \citenamefont {{van der Walt}}, \citenamefont {Brett}, \citenamefont {Wilson}, \citenamefont {Millman}, \citenamefont {Mayorov}, \citenamefont {Nelson}, \citenamefont {Jones}, \citenamefont {Kern}, \citenamefont {Larson}, \citenamefont {Carey}, \citenamefont {Polat}, \citenamefont {Feng}, \citenamefont {Moore}, \citenamefont {{VanderPlas}}, \citenamefont {Laxalde}, \citenamefont {Perktold}, \citenamefont {Cimrman}, \citenamefont {Henriksen}, \citenamefont {Quintero}, \citenamefont {Harris}, \citenamefont {Archibald}, \citenamefont {Ribeiro}, \citenamefont {Pedregosa}, \citenamefont {{van Mulbregt}},\ and\ \citenamefont {{SciPy 1.0 Contributors}}}]{2020SciPy-NMeth}%
  \BibitemOpen
  \bibfield  {author} {\bibinfo {author} {\bibfnamefont {P.}~\bibnamefont {Virtanen}}, \bibinfo {author} {\bibfnamefont {R.}~\bibnamefont {Gommers}}, \bibinfo {author} {\bibfnamefont {T.~E.}\ \bibnamefont {Oliphant}}, \bibinfo {author} {\bibfnamefont {M.}~\bibnamefont {Haberland}}, \bibinfo {author} {\bibfnamefont {T.}~\bibnamefont {Reddy}}, \bibinfo {author} {\bibfnamefont {D.}~\bibnamefont {Cournapeau}}, \bibinfo {author} {\bibfnamefont {E.}~\bibnamefont {Burovski}}, \bibinfo {author} {\bibfnamefont {P.}~\bibnamefont {Peterson}}, \bibinfo {author} {\bibfnamefont {W.}~\bibnamefont {Weckesser}}, \bibinfo {author} {\bibfnamefont {J.}~\bibnamefont {Bright}}, \bibinfo {author} {\bibfnamefont {S.~J.}\ \bibnamefont {{van der Walt}}}, \bibinfo {author} {\bibfnamefont {M.}~\bibnamefont {Brett}}, \bibinfo {author} {\bibfnamefont {J.}~\bibnamefont {Wilson}}, \bibinfo {author} {\bibfnamefont {K.~J.}\ \bibnamefont {Millman}}, \bibinfo {author} {\bibfnamefont {N.}~\bibnamefont {Mayorov}}, \bibinfo {author} {\bibfnamefont
  {A.~R.~J.}\ \bibnamefont {Nelson}}, \bibinfo {author} {\bibfnamefont {E.}~\bibnamefont {Jones}}, \bibinfo {author} {\bibfnamefont {R.}~\bibnamefont {Kern}}, \bibinfo {author} {\bibfnamefont {E.}~\bibnamefont {Larson}}, \bibinfo {author} {\bibfnamefont {C.~J.}\ \bibnamefont {Carey}}, \bibinfo {author} {\bibfnamefont {{\.I}.}~\bibnamefont {Polat}}, \bibinfo {author} {\bibfnamefont {Y.}~\bibnamefont {Feng}}, \bibinfo {author} {\bibfnamefont {E.~W.}\ \bibnamefont {Moore}}, \bibinfo {author} {\bibfnamefont {J.}~\bibnamefont {{VanderPlas}}}, \bibinfo {author} {\bibfnamefont {D.}~\bibnamefont {Laxalde}}, \bibinfo {author} {\bibfnamefont {J.}~\bibnamefont {Perktold}}, \bibinfo {author} {\bibfnamefont {R.}~\bibnamefont {Cimrman}}, \bibinfo {author} {\bibfnamefont {I.}~\bibnamefont {Henriksen}}, \bibinfo {author} {\bibfnamefont {E.~A.}\ \bibnamefont {Quintero}}, \bibinfo {author} {\bibfnamefont {C.~R.}\ \bibnamefont {Harris}}, \bibinfo {author} {\bibfnamefont {A.~M.}\ \bibnamefont {Archibald}}, \bibinfo {author}
  {\bibfnamefont {A.~H.}\ \bibnamefont {Ribeiro}}, \bibinfo {author} {\bibfnamefont {F.}~\bibnamefont {Pedregosa}}, \bibinfo {author} {\bibfnamefont {P.}~\bibnamefont {{van Mulbregt}}},\ and\ \bibinfo {author} {\bibnamefont {{SciPy 1.0 Contributors}}},\ }\bibfield  {title} {\bibinfo {title} {{{SciPy} 1.0: Fundamental Algorithms for Scientific Computing in Python}},\ }\href {https://doi.org/10.1038/s41592-019-0686-2} {\bibfield  {journal} {\bibinfo  {journal} {Nature Methods}\ }\textbf {\bibinfo {volume} {17}},\ \bibinfo {pages} {261} (\bibinfo {year} {2020})}\BibitemShut {NoStop}%
\bibitem [{\citenamefont {Harris}\ \emph {et~al.}(2020)\citenamefont {Harris}, \citenamefont {Millman}, \citenamefont {van~der Walt}, \citenamefont {Gommers}, \citenamefont {Virtanen}, \citenamefont {Cournapeau}, \citenamefont {Wieser}, \citenamefont {Taylor}, \citenamefont {Berg}, \citenamefont {Smith}, \citenamefont {Kern}, \citenamefont {Picus}, \citenamefont {Hoyer}, \citenamefont {van Kerkwijk}, \citenamefont {Brett}, \citenamefont {Haldane}, \citenamefont {del R{\'{i}}o}, \citenamefont {Wiebe}, \citenamefont {Peterson}, \citenamefont {G{\'{e}}rard-Marchant}, \citenamefont {Sheppard}, \citenamefont {Reddy}, \citenamefont {Weckesser}, \citenamefont {Abbasi}, \citenamefont {Gohlke},\ and\ \citenamefont {Oliphant}}]{numpy}%
  \BibitemOpen
  \bibfield  {author} {\bibinfo {author} {\bibfnamefont {C.~R.}\ \bibnamefont {Harris}}, \bibinfo {author} {\bibfnamefont {K.~J.}\ \bibnamefont {Millman}}, \bibinfo {author} {\bibfnamefont {S.~J.}\ \bibnamefont {van~der Walt}}, \bibinfo {author} {\bibfnamefont {R.}~\bibnamefont {Gommers}}, \bibinfo {author} {\bibfnamefont {P.}~\bibnamefont {Virtanen}}, \bibinfo {author} {\bibfnamefont {D.}~\bibnamefont {Cournapeau}}, \bibinfo {author} {\bibfnamefont {E.}~\bibnamefont {Wieser}}, \bibinfo {author} {\bibfnamefont {J.}~\bibnamefont {Taylor}}, \bibinfo {author} {\bibfnamefont {S.}~\bibnamefont {Berg}}, \bibinfo {author} {\bibfnamefont {N.~J.}\ \bibnamefont {Smith}}, \bibinfo {author} {\bibfnamefont {R.}~\bibnamefont {Kern}}, \bibinfo {author} {\bibfnamefont {M.}~\bibnamefont {Picus}}, \bibinfo {author} {\bibfnamefont {S.}~\bibnamefont {Hoyer}}, \bibinfo {author} {\bibfnamefont {M.~H.}\ \bibnamefont {van Kerkwijk}}, \bibinfo {author} {\bibfnamefont {M.}~\bibnamefont {Brett}}, \bibinfo {author} {\bibfnamefont
  {A.}~\bibnamefont {Haldane}}, \bibinfo {author} {\bibfnamefont {J.~F.}\ \bibnamefont {del R{\'{i}}o}}, \bibinfo {author} {\bibfnamefont {M.}~\bibnamefont {Wiebe}}, \bibinfo {author} {\bibfnamefont {P.}~\bibnamefont {Peterson}}, \bibinfo {author} {\bibfnamefont {P.}~\bibnamefont {G{\'{e}}rard-Marchant}}, \bibinfo {author} {\bibfnamefont {K.}~\bibnamefont {Sheppard}}, \bibinfo {author} {\bibfnamefont {T.}~\bibnamefont {Reddy}}, \bibinfo {author} {\bibfnamefont {W.}~\bibnamefont {Weckesser}}, \bibinfo {author} {\bibfnamefont {H.}~\bibnamefont {Abbasi}}, \bibinfo {author} {\bibfnamefont {C.}~\bibnamefont {Gohlke}},\ and\ \bibinfo {author} {\bibfnamefont {T.~E.}\ \bibnamefont {Oliphant}},\ }\bibfield  {title} {\bibinfo {title} {Array programming with {NumPy}},\ }\href {https://doi.org/10.1038/s41586-020-2649-2} {\bibfield  {journal} {\bibinfo  {journal} {Nature}\ }\textbf {\bibinfo {volume} {585}},\ \bibinfo {pages} {357} (\bibinfo {year} {2020})}\BibitemShut {NoStop}%
\bibitem [{\citenamefont {Lam}\ \emph {et~al.}(2015)\citenamefont {Lam}, \citenamefont {Pitrou},\ and\ \citenamefont {Seibert}}]{numba}%
  \BibitemOpen
  \bibfield  {author} {\bibinfo {author} {\bibfnamefont {S.~K.}\ \bibnamefont {Lam}}, \bibinfo {author} {\bibfnamefont {A.}~\bibnamefont {Pitrou}},\ and\ \bibinfo {author} {\bibfnamefont {S.}~\bibnamefont {Seibert}},\ }\bibfield  {title} {\bibinfo {title} {Numba: a llvm-based python jit compiler},\ }in\ \href {https://doi.org/10.1145/2833157.2833162} {\emph {\bibinfo {booktitle} {Proceedings of the Second Workshop on the LLVM Compiler Infrastructure in HPC}}},\ \bibinfo {series and number} {LLVM '15}\ (\bibinfo  {publisher} {Association for Computing Machinery},\ \bibinfo {address} {New York, NY, USA},\ \bibinfo {year} {2015})\BibitemShut {NoStop}%
\bibitem [{\citenamefont {Marso}\ \emph {et~al.}()\citenamefont {Marso}, \citenamefont {Herbst}, \citenamefont {Wilkens}, \citenamefont {De~Maio}, \citenamefont {Brandic},\ and\ \citenamefont {Küng}}]{GitHub}%
  \BibitemOpen
  \bibfield  {author} {\bibinfo {author} {\bibfnamefont {M.}~\bibnamefont {Marso}}, \bibinfo {author} {\bibfnamefont {S.}~\bibnamefont {Herbst}}, \bibinfo {author} {\bibfnamefont {J.}~\bibnamefont {Wilkens}}, \bibinfo {author} {\bibfnamefont {V.}~\bibnamefont {De~Maio}}, \bibinfo {author} {\bibfnamefont {I.}~\bibnamefont {Brandic}},\ and\ \bibinfo {author} {\bibfnamefont {R.}~\bibnamefont {Küng}},\ }\href {https://github.com/sabrinaherbst/online-entanglement-verification} {\bibinfo {title} {Source code for "an online approach for entanglement verification using classical shadows"}}\BibitemShut {NoStop}%
\bibitem [{\citenamefont {Elben}\ and\ \citenamefont {Vermersch}(2025)}]{elben2025randommeasjljuliapackagerandomized}%
  \BibitemOpen
  \bibfield  {author} {\bibinfo {author} {\bibfnamefont {A.}~\bibnamefont {Elben}}\ and\ \bibinfo {author} {\bibfnamefont {B.}~\bibnamefont {Vermersch}},\ }\href {https://arxiv.org/abs/2509.12749} {\bibinfo {title} {Randommeas.jl: A julia package for randomized measurements in quantum devices}} (\bibinfo {year} {2025}),\ \Eprint {https://arxiv.org/abs/2509.12749} {arXiv:2509.12749 [quant-ph]} \BibitemShut {NoStop}%
\bibitem [{\citenamefont {Burden}\ and\ \citenamefont {Faires}(1997)}]{burden1997numerical}%
  \BibitemOpen
  \bibfield  {author} {\bibinfo {author} {\bibfnamefont {R.~L.}\ \bibnamefont {Burden}}\ and\ \bibinfo {author} {\bibfnamefont {J.~D.}\ \bibnamefont {Faires}},\ }\href@noop {} {\bibinfo {title} {Numerical analysis, brooks}} (\bibinfo {year} {1997})\BibitemShut {NoStop}%
\bibitem [{\citenamefont {Huang}\ \emph {et~al.}(2021)\citenamefont {Huang}, \citenamefont {Kueng},\ and\ \citenamefont {Preskill}}]{PhysRevLett.126.190505}%
  \BibitemOpen
  \bibfield  {author} {\bibinfo {author} {\bibfnamefont {H.-Y.}\ \bibnamefont {Huang}}, \bibinfo {author} {\bibfnamefont {R.}~\bibnamefont {Kueng}},\ and\ \bibinfo {author} {\bibfnamefont {J.}~\bibnamefont {Preskill}},\ }\bibfield  {title} {\bibinfo {title} {Information-theoretic bounds on quantum advantage in machine learning},\ }\href {https://doi.org/10.1103/PhysRevLett.126.190505} {\bibfield  {journal} {\bibinfo  {journal} {Phys. Rev. Lett.}\ }\textbf {\bibinfo {volume} {126}},\ \bibinfo {pages} {190505} (\bibinfo {year} {2021})}\BibitemShut {NoStop}%
\bibitem [{\citenamefont {Huang}\ \emph {et~al.}(2022)\citenamefont {Huang}, \citenamefont {Broughton}, \citenamefont {Cotler}, \citenamefont {Chen}, \citenamefont {Li}, \citenamefont {Mohseni}, \citenamefont {Neven}, \citenamefont {Babbush}, \citenamefont {Kueng}, \citenamefont {Preskill},\ and\ \citenamefont {McClean}}]{10.1126/science.abn7293}%
  \BibitemOpen
  \bibfield  {author} {\bibinfo {author} {\bibfnamefont {H.-Y.}\ \bibnamefont {Huang}}, \bibinfo {author} {\bibfnamefont {M.}~\bibnamefont {Broughton}}, \bibinfo {author} {\bibfnamefont {J.}~\bibnamefont {Cotler}}, \bibinfo {author} {\bibfnamefont {S.}~\bibnamefont {Chen}}, \bibinfo {author} {\bibfnamefont {J.}~\bibnamefont {Li}}, \bibinfo {author} {\bibfnamefont {M.}~\bibnamefont {Mohseni}}, \bibinfo {author} {\bibfnamefont {H.}~\bibnamefont {Neven}}, \bibinfo {author} {\bibfnamefont {R.}~\bibnamefont {Babbush}}, \bibinfo {author} {\bibfnamefont {R.}~\bibnamefont {Kueng}}, \bibinfo {author} {\bibfnamefont {J.}~\bibnamefont {Preskill}},\ and\ \bibinfo {author} {\bibfnamefont {J.~R.}\ \bibnamefont {McClean}},\ }\bibfield  {title} {\bibinfo {title} {Quantum advantage in learning from experiments},\ }\href {https://doi.org/10.1126/science.abn7293} {\bibfield  {journal} {\bibinfo  {journal} {Science}\ }\textbf {\bibinfo {volume} {376}},\ \bibinfo {pages} {1182} (\bibinfo {year} {2022})},\ \Eprint
  {https://arxiv.org/abs/https://www.science.org/doi/pdf/10.1126/science.abn7293} {https://www.science.org/doi/pdf/10.1126/science.abn7293} \BibitemShut {NoStop}%
\bibitem [{\citenamefont {N{\"o}ller}\ \emph {et~al.}(2025)\citenamefont {N{\"o}ller}, \citenamefont {Tran}, \citenamefont {Gachechiladze},\ and\ \citenamefont {Kueng}}]{noller2025infinite}%
  \BibitemOpen
  \bibfield  {author} {\bibinfo {author} {\bibfnamefont {J.}~\bibnamefont {N{\"o}ller}}, \bibinfo {author} {\bibfnamefont {V.~T.}\ \bibnamefont {Tran}}, \bibinfo {author} {\bibfnamefont {M.}~\bibnamefont {Gachechiladze}},\ and\ \bibinfo {author} {\bibfnamefont {R.}~\bibnamefont {Kueng}},\ }\bibfield  {title} {\bibinfo {title} {An infinite hierarchy of multi-copy quantum learning tasks},\ }\href@noop {} {\bibfield  {journal} {\bibinfo  {journal} {arXiv preprint arXiv:2510.08070}\ } (\bibinfo {year} {2025})}\BibitemShut {NoStop}%
\end{thebibliography}%

\onecolumngrid
\appendix

\section{Combinatorial closed form for the m-th moment from Pauli classical shadows}
\label{app:combinatorial-moment}

We derive an explicit closed-form expression for the estimator of the $m$-th moment
\[
\Tr(\rho^m)
\]
constructed from Pauli classical shadows. The derivation makes the underlying combinatorial structure fully explicit and does not rely on kernel formulations or auxiliary abstractions.

\subsection*{Classical shadow snapshots}

Each measurement shot $t \in \{1,\dots,T\}$ produces a classical snapshot
\begin{equation}
\hat{\rho}_t
=
\bigotimes_{n=1}^N \hat{\rho}_{t,n},
\qquad
\hat{\rho}_{t,n}
=
3\,U_{t,n}^\dagger \ketbra{b_{t,n}}{b_{t,n}} U_{t,n}
-
\mathbb{I},
\label{eq:single-shot-shadow}
\end{equation}
where $U_{t,n} \in \{\mathbb{I},H,S^\dagger H\}$ rotates the computational basis into a Pauli eigenbasis and $b_{t,n}\in\{0,1\}$ is the measurement outcome.

Since $U_{t,n}$ maps computational basis states to eigenstates of a Pauli operator, the rotated projector admits the form
\begin{equation}
U_{t,n}^\dagger \ketbra{b_{t,n}}{b_{t,n}} U_{t,n}
=
\frac{1}{2}\bigl(\mathbb{I} + s_{t,n} P_{t,n}\bigr),
\qquad
s_{t,n}=(-1)^{b_{t,n}},
\end{equation}
where $P_{t,n}\in\{X,Y,Z\}$ denotes the measured Pauli axis. Substituting yields the explicit affine form
\begin{equation}
\hat{\rho}_{t,n}
=
\frac{1}{2}\mathbb{I}
+
\frac{3}{2}\, s_{t,n} P_{t,n}.
\label{eq:affine-single-qubit}
\end{equation}

\subsection*{Moment estimator and factorization}

For any fixed $m$-tuple of \emph{distinct} shots $t_1 < \dots < t_m$, the moment estimator involves
\[
\Tr\!\left(\hat{\rho}_{t_1}\cdots \hat{\rho}_{t_m}\right).
\]
Using the tensor-product structure of the snapshots,
\begin{equation}
\hat{\rho}_{t_1}\cdots \hat{\rho}_{t_m}
=
\bigotimes_{n=1}^N
\left(
\hat{\rho}_{t_1,n}\cdots \hat{\rho}_{t_m,n}
\right),
\end{equation}
and hence the trace factorizes as
\begin{equation}
\Tr\!\left(\hat{\rho}_{t_1}\cdots \hat{\rho}_{t_m}\right)
=
\prod_{n=1}^N
\Tr\!\left(
\hat{\rho}_{t_1,n}\cdots \hat{\rho}_{t_m,n}
\right),
\label{eq:trace-factorization}
\end{equation}
where $\tr(\cdot)$ denotes the single-qubit trace. The problem therefore reduces to a single-qubit calculation.

\subsection*{Subset expansion at the single-qubit level}

Fix a qubit $n$. Using Eq.\eqref{eq:affine-single-qubit}, we expand
\begin{equation}
\prod_{j=1}^m
\hat{\rho}_{t_j,n}
=
\prod_{j=1}^m
\left(
\frac{1}{2}\mathbb{I}
+
\frac{3}{2}\, s_{t_j,n} P_{t_j,n}
\right).
\end{equation}
This product is expanded by choosing, for each index $j$, either the identity term or the Pauli term. Equivalently, the expansion is indexed by subsets
$S \subseteq \{1,\dots,m\}$, where indices in $S$ select the Pauli term and indices not in $S$ select the identity term. This yields
\begin{equation}
\prod_{j=1}^m
\hat{\rho}_{t_j,n}
=
\sum_{S\subseteq\{1,\dots,m\}}
\left(\frac{1}{2}\right)^{m-|S|}
\left(\frac{3}{2}\right)^{|S|}
\left(
\prod_{j\in S} s_{t_j,n}
\right)
\left(
\prod_{j\in S} P_{t_j,n}
\right),
\label{eq:subset-expansion}
\end{equation}
where the Pauli product is taken in the original shot order.

Taking the single-qubit trace gives
\begin{equation}
\tr\!\left(
\hat{\rho}_{t_1,n}\cdots \hat{\rho}_{t_m,n}
\right)
=
\sum_{S\subseteq\{1,\dots,m\}}
\left(\frac{1}{2}\right)^{m-|S|}
\left(\frac{3}{2}\right)^{|S|}
\left(
\prod_{j\in S} s_{t_j,n}
\right)
\tr\!\left(
\prod_{j\in S} P_{t_j,n}
\right).
\label{eq:single-qubit-trace}
\end{equation}

\subsection*{Pauli trace selection rule}

For a single qubit,
\[
\tr(\mathbb{I}) = 2,
\qquad
\tr(X) = \tr(Y) = \tr(Z) = 0.
\]
Consequently, the trace
$\tr\!\left(\prod_{j\in S} P_{t_j,n}\right)$
is nonzero if and only if the Pauli product equals $\pm \mathbb{I}$. In this case,
\[
\tr\!\left(\prod_{j\in S} P_{t_j,n}\right) = \pm 2,
\]
with the sign determined by the Pauli multiplication rules. All other subsets give zero contribution. Thus, only subsets whose Pauli labels multiply to the identity survive.

\subsection*{Final closed combinatorial form}

Combining the previous steps, for any fixed $m$-tuple $t_1<\dots<t_m$,
\begin{equation}
\Tr\!\left(\hat{\rho}_{t_1}\cdots \hat{\rho}_{t_m}\right)
=
\prod_{n=1}^N
\left[
\sum_{S\subseteq\{1,\dots,m\}}
\left(\frac{1}{2}\right)^{m-|S|}
\left(\frac{3}{2}\right)^{|S|}
\left(
\prod_{j\in S} s_{t_j,n}
\right)
\tr\!\left(
\prod_{j\in S} P_{t_j,n}
\right)
\right].
\label{eq:final-combinatorial-form}
\end{equation}

The estimator of $\Tr(\rho^m)$ is obtained by uniformly averaging \eqref{eq:final-combinatorial-form} over all $\binom{T}{m}$ distinct $m$-tuples of shots.

\subsection*{Remark on cyclicity of the trace}

The cyclicity of the trace guarantees invariance of
\[
\tr\!\left(
P_{t_{j_1},n}\cdots P_{t_{j_r},n}
\right)
\]
under cyclic permutations of the Pauli product. The overall sign $\pm 2$ is fixed by the Pauli multiplication rules but is independent of where the product is cyclically started.

\subsection*{Remark on partial transpose}

Since $\tr(A^T)=\tr(A)$, the same combinatorial expansion applies to
\[
\Tr\!\left(
\hat{\rho}_{t_1}^{T_B}\cdots \hat{\rho}_{t_m}^{T_B}
\right).
\]
The only modification is a sign flip for each $Y$ operator acting on qubits belonging to subsystem $B$.

\section{Partial transpose of a single-shot classical shadow on subsystem \texorpdfstring{$B$}{B}}
\label{app:pt_shadow}

We consider an $N$-qubit single-shot classical shadow snapshot of the form
\begin{equation}
    \hat{\rho} \;=\; \bigotimes_{n=1}^N \hat{\rho}_{n}, 
    \qquad 
    \hat{\rho}_{n} \;=\; 3\, U_n^\dagger \ketbra{b_n}{b_n} U_n \;-\; \mathbb{I},
    \label{single_shot_shadow_app}
\end{equation}
where $b_n\in\{0,1\}$ is the computational-basis measurement outcome after applying a single-qubit unitary $U_n$.
In standard local classical shadows, each $U_n$ is chosen from the single-qubit Clifford set that maps the computational $Z$-basis measurement to a random Pauli basis measurement. Concretely, one may take
\begin{equation}
    U_n \in \{ \mathbb{I},\, H,\, SH\},
\end{equation}
which correspond to measuring in the $Z$, $X$, and $Y$ eigenbases, respectively.

Let $B\subseteq \{1,\dots,N\}$ denote the subsystem to be partially transposed.
We write $T_B$ for partial transpose on $B$ in the computational basis, i.e.
\begin{equation}
    (X_A\otimes Y_B)^{T_B} \;=\; X_A \otimes Y_B^{T},
\end{equation}
extended by linearity.

\subsection{Tensor-factorization of the partial transpose}
\begin{lemma}
\label{lem:pt_tensor}
For any product operator $\hat{\rho}=\bigotimes_{n=1}^N \hat{\rho}_n$ and any subset $B$,
\begin{equation}
    \hat{\rho}^{T_B}
    \;=\;
    \bigotimes_{n\notin B}\hat{\rho}_n
    \;\otimes\;
    \bigotimes_{n\in B}\hat{\rho}_n^{T}.
\end{equation}
\end{lemma}
\begin{proof}
It suffices to apply $(X_A\otimes Y_B)^{T_B}=X_A\otimes Y_B^T$ to a simple tensor and extend by multilinearity across tensor products of single-qubit factors.
\end{proof}

Thus, the problem reduces to computing $\hat{\rho}_n^T$ for a \emph{single qubit}.
Since $\mathbb{I}^T=\mathbb{I}$, we only need the transpose of the rank-one term:
\begin{equation}
    \hat{\rho}_n^T
    \;=\;
    3\left(U_n^\dagger \ketbra{b_n}{b_n} U_n\right)^T \;-\; \mathbb{I}.
\end{equation}

\subsection{Single-qubit transpose action: explicit derivation}
Define the post-rotation pure state
\begin{equation}
    \ket{\psi_{U,b}} \;:=\; U^\dagger \ket{b},
    \qquad
    \Pi_{U,b} \;:=\; \ketbra{\psi_{U,b}}{\psi_{U,b}}
    \;=\;
    U^\dagger \ketbra{b}{b} U .
\end{equation}
Then
\begin{equation}
    \Pi_{U,b}^T \;=\; \left(\ketbra{\psi_{U,b}}{\psi_{U,b}}\right)^T
    \;=\;
    \ketbra{\psi_{U,b}^*}{\psi_{U,b}^*},
\end{equation}
where ${}^*$ denotes entrywise complex conjugation in the computational basis.

We now evaluate $\ket{\psi_{U,b}^*}$ for the three relevant choices of $U$.

\paragraph{Case 1: $U=\mathbb{I}$ (Pauli-$Z$ basis).}
Here $\ket{\psi_{\mathbb{I},b}}=\ket{b}$ has real amplitudes, hence $\ket{\psi_{\mathbb{I},b}^*}=\ket{b}$ and
\begin{equation}
    \Pi_{\mathbb{I},b}^T \;=\; \Pi_{\mathbb{I},b}.
\end{equation}

\paragraph{Case 2: $U=H$ (Pauli-$X$ basis).}
Using $H=H^T$ and that $H$ is real, we have $\ket{\psi_{H,b}}=H\ket{b}$ with real amplitudes, hence again
\begin{equation}
    \Pi_{H,b}^T \;=\; \Pi_{H,b}.
\end{equation}

\paragraph{Case 3: $U=SH$ (Pauli-$Y$ basis).}
Write the $Y$-eigenstates as
\begin{equation}
    \ket{+_y}=\frac{1}{\sqrt{2}}(\ket{0}+i\ket{1}),
    \qquad
    \ket{-_y}=\frac{1}{\sqrt{2}}(\ket{0}-i\ket{1}).
\end{equation}
One checks that $(SH)^\dagger\ket{0}=\ket{+_y}$ and $(SH)^\dagger\ket{1}=\ket{-_y}$, i.e.
\begin{equation}
    \ket{\psi_{SH,0}}=\ket{+_y},
    \qquad
    \ket{\psi_{SH,1}}=\ket{-_y}.
\end{equation}
Complex conjugation swaps these states:
\begin{equation}
    \ket{+_y}^* \;=\; \frac{1}{\sqrt{2}}(\ket{0}-i\ket{1}) \;=\; \ket{-_y},
    \qquad
    \ket{-_y}^* \;=\; \ket{+_y}.
\end{equation}
Therefore the corresponding projectors are exchanged under transpose:
\begin{equation}
    \Pi_{SH,0}^T \;=\; \Pi_{SH,1},
    \qquad
    \Pi_{SH,1}^T \;=\; \Pi_{SH,0}.
\end{equation}
Equivalently, in the $Y$-basis case,
\begin{equation}
    \Pi_{SH,b}^T \;=\; \Pi_{SH,\,1-b}.
\end{equation}

\subsection{Closed form for the transposed single-qubit shadow factor}
Combining all cases, for $U\in\{\mathbb{I},H,SH\}$ we have
\begin{equation}
    \left(U^\dagger \ketbra{b}{b} U\right)^T
    \;=\;
    U^\dagger \ketbra{\tilde b(U,b)}{\tilde b(U,b)} U,
\end{equation}
where the transformed bit $\tilde b(U,b)$ is given by
\begin{equation}
    \tilde b(U,b)
    \;=\;
    \begin{cases}
      b, & U\in\{\mathbb{I},H\} \quad (Z \text{ or } X\text{ basis}),\\
      1-b, & U=SH \quad (Y\text{ basis}).
    \end{cases}
\end{equation}
Hence the transposed single-qubit shadow factor is
\begin{equation}
    \hat{\rho}_n^{T}
    \;=\;
    3\,U_n^\dagger \ketbra{\tilde b_n}{\tilde b_n} U_n \;-\;\mathbb{I},
    \qquad
    \tilde b_n := \tilde b(U_n,b_n).
    \label{eq:single_qubit_pt_shadow}
\end{equation}

\subsection{Final expression: partial transpose on an arbitrary subsystem \texorpdfstring{$B$}{B}}
Using Lemma~\ref{lem:pt_tensor} and \eqref{eq:single_qubit_pt_shadow}, the partial transpose of the full snapshot
\eqref{single_shot_shadow_app} on subsystem $B$ is
\begin{equation}
    \hat{\rho}^{T_B}
    \;=\;
    \bigotimes_{n\notin B}\Bigl(3\,U_n^\dagger \ketbra{b_n}{b_n} U_n - \mathbb{I}\Bigr)
    \;\otimes\;
    \bigotimes_{n\in B}\Bigl(3\,U_n^\dagger \ketbra{\tilde b_n}{\tilde b_n} U_n - \mathbb{I}\Bigr),
    \label{eq:ptB_full_snapshot}
\end{equation}
with $\tilde b_n=b_n$ for $U_n\in\{\mathbb{I},H\}$ and $\tilde b_n=1-b_n$ for $U_n=SH$.
In words: \emph{partial transpose on $B$ leaves $Z$- and $X$-basis shadow factors unchanged, and flips the recorded bit
$b_n\mapsto 1-b_n$ exactly on those qubits in $B$ that were measured in the $Y$ basis.}

\section{Online Update of the Estimator}
\label{app:streaming-update}

Recall the unbiased estimator
\begin{equation}
    \widehat{P}^{(m)}_T
    := \frac{1}{\binom{T}{m}}
       \sum_{1 \le t_1 < \cdots < t_m \le T}
       \mathrm{Tr}\!\left(
           \hat{\rho}_{t_1}^{T_B}
           \cdots
           \hat{\rho}_{t_m}^{T_B}
       \right),
    \label{eq:app-pt-moment-shadow-estimator}
\end{equation}
defined for integers $T\ge m$.

\subsection{Update Formula and Proof}

\begin{theorem}[Online update]
\label{thm:update-formula}
For every $T\ge m$,
\begin{equation}
  \widehat{P}^{(m)}_{T+1}
  =
  \Bigl(1-\frac{m}{T+1}\Bigr)\widehat{P}^{(m)}_{T}
  \;+\;
  \frac{m}{T+1}\;
  \frac{1}{\binom{T}{m-1}}
  \sum_{1\le t_1<\cdots<t_{m-1}\le T}
  \Tr\!\left(
    \hat\rho_{t_1}^{T_B}\cdots \hat\rho_{t_{m-1}}^{T_B}\,\hat\rho_{T+1}^{T_B}
  \right).
  \label{eq:thm-streaming-update}
\end{equation}
\end{theorem}

\begin{proof}
Multiply \eqref{eq:app-pt-moment-shadow-estimator} at time $(T+1)$ by $\binom{T+1}{m}$:
\[
\binom{T+1}{m}\,\widehat{P}^{(m)}_{T+1}
=
\sum_{1\le t_1<\cdots<t_m\le T+1}
\Tr\!\left(
\hat\rho_{t_1}^{T_B}\cdots \hat\rho_{t_m}^{T_B}
\right).
\]
Partition the sum into $m$-tuples that exclude $(T+1)$ and those that include $(T+1)$:
\begin{align*}
\binom{T+1}{m}\,\widehat{P}^{(m)}_{T+1}
&=
\sum_{1\le t_1<\cdots<t_m\le T}
\Tr\!\left(
\hat\rho_{t_1}^{T_B}\cdots \hat\rho_{t_m}^{T_B}
\right)\\
&\quad+
\sum_{1\le t_1<\cdots<t_{m-1}\le T}
\Tr\!\left(
\hat\rho_{t_1}^{T_B}\cdots \hat\rho_{t_{m-1}}^{T_B}\hat\rho_{T+1}^{T_B}
\right).
\end{align*}
The first sum equals $\binom{T}{m}\widehat{P}^{(m)}_{T}$ by \eqref{eq:app-pt-moment-shadow-estimator}. Therefore,
\[
\binom{T+1}{m}\,\widehat{P}^{(m)}_{T+1}
=
\binom{T}{m}\widehat{P}^{(m)}_{T}
+
\sum_{1\le t_1<\cdots<t_{m-1}\le T}
\Tr\!\left(
\hat\rho_{t_1}^{T_B}\cdots \hat\rho_{t_{m-1}}^{T_B}\hat\rho_{T+1}^{T_B}
\right).
\]
Divide by $\binom{T+1}{m}$ and use Pascal's identity
$\binom{T+1}{m}=\binom{T}{m}+\binom{T}{m-1}$ to rewrite the coefficients:
\[
\frac{\binom{T}{m}}{\binom{T+1}{m}}=\frac{T+1-m}{T+1}=1-\frac{m}{T+1},
\qquad
\frac{1}{\binom{T+1}{m}}
=
\frac{\binom{T}{m-1}}{\binom{T+1}{m}}\cdot\frac{1}{\binom{T}{m-1}}
=
\frac{m}{T+1}\cdot\frac{1}{\binom{T}{m-1}}.
\]
Substituting these identities yields \eqref{eq:thm-streaming-update}.
\end{proof}

\section{Descartes' rule of signs and the ESP--moment hierarchy}
\label{app:ng}
We consider the partial transpose $A:=\rho^{T_B}\in\mathbb{C}^{d\times d}$, which is Hermitian and
therefore has real eigenvalues $\lambda_1,\dots,\lambda_d\in\mathbb{R}$. NPT entanglement is
equivalent to $A$ having at least one negative eigenvalue. Our entanglement test proceeds by
(i) estimating the power sums $p_m:=\Tr(A^m)$ from randomized measurements, (ii) converting
$(p_1,\dots,p_k)$ into elementary symmetric polynomials $e_k(A)$ via Newton--Girard, and
(iii) checking the sign constraints implied by $A\succeq 0$. Descartes' rule provides an
additional interpretation: the pattern of signs of $(1,e_1,\dots,e_d)$ upper-bounds (and fixes
the parity of) the number of negative eigenvalues of $A$.

\paragraph{Elementary symmetric polynomials.}
Define
\[
e_k(A) \;:=\; e_k(\lambda_1,\dots,\lambda_d)
= \sum_{1\le i_1<\cdots<i_k\le d}\lambda_{i_1}\cdots\lambda_{i_k},
\qquad k=1,\dots,d,
\]
with $e_0(A):=1$.

\paragraph{A polynomial whose positive roots are the negative eigenvalues.}
The characteristic polynomial is
\begin{equation}
\chi_A(x) := \det(x\mathbb{I}-A)
= x^d - e_1(A)\,x^{d-1} + e_2(A)\,x^{d-2} - \cdots + (-1)^d e_d(A).
\label{eq:charpoly-esp}
\end{equation}
Define
\begin{equation}
q_A(x) := \det(x\mathbb{I}+A)=(-1)^d\,\chi_A(-x)
= x^d + e_1(A)\,x^{d-1} + e_2(A)\,x^{d-2} + \cdots + e_d(A).
\label{eq:qA-def}
\end{equation}
Since $q_A(x)=\prod_{i=1}^d (x+\lambda_i)$, its positive real roots are exactly the numbers
$x=-\lambda_i>0$ corresponding to negative eigenvalues $\lambda_i<0$ (with multiplicities).

\paragraph{Sign variation.}
For a real sequence $(a_0,\dots,a_d)$, let $V(a_0,\dots,a_d)$ denote the number of sign changes
after deleting zero entries.

\begin{lemma}[Descartes' rule $\Rightarrow$ bound on negative eigenvalues]
\label{lem:descartes-npt}
Let $A=\rho^{T_B}$ and $q_A$ be as in \eqref{eq:qA-def}. Then the number of negative eigenvalues
of $A$ (counted with multiplicity) equals the number of positive real roots of $q_A$, and hence
\begin{align}
\#\{i:\lambda_i<0\} \;&\le\; V\bigl(1,\,e_1(A),\,e_2(A),\,\dots,\,e_d(A)\bigr), \label{eq:descartes-bound}\\
\#\{i:\lambda_i<0\} \;&\equiv\; V\bigl(1,\,e_1(A),\,e_2(A),\,\dots,\,e_d(A)\bigr)\pmod 2. \label{eq:descartes-parity}
\end{align}
In particular, if $e_k(A)\ge 0$ for all $k$, then $q_A$ has no positive roots and $A\succeq 0$.
Conversely, if $e_k(A)<0$ for some $k$, then necessarily $A$ has a negative eigenvalue, i.e.\ $\rho$ is NPT entangled.
\end{lemma}

\begin{proof}
The factorization $q_A(x)=\prod_{i=1}^d(x+\lambda_i)$ shows that positive roots of $q_A$ are in one-to-one
correspondence with negative eigenvalues of $A$, preserving multiplicities. Descartes' rule of signs applied
to $q_A$ yields \eqref{eq:descartes-bound} and \eqref{eq:descartes-parity}. The final implications follow.
\end{proof}

\paragraph{Newton--Girard recurrence.}
The power sums $p_k:=\Tr(A^k)=\sum_i\lambda_i^k$ and elementary symmetric polynomials $e_k$ are
related by the Newton--Girard identities: with $e_0:=1$,
\begin{equation}
  \label{eq:newton-girard-werner}
  k\,e_k \;=\; \sum_{j=1}^{k}(-1)^{j-1}\,p_j\,e_{k-j},
  \qquad k\ge 1.
\end{equation}
This allows one to express each $e_k$ as a polynomial in $p_1,\dots,p_k$, enabling the
PPT conditions $e_k\ge 0$ to be checked directly from estimated moments.

\paragraph{Low-order moment constraints (isolating the new moment).}
Let
\[
p_m:=\Tr\!\bigl[(\rho_W^{T_B}(t))^m\bigr],\qquad m\ge 1,
\]
so $p_1=\Tr(\rho_W^{T_B}(t))=1$.
Eliminating $e_1,e_2,\dots$ recursively via \eqref{eq:newton-girard-werner}, the first nontrivial conditions
$e_k\ge 0$ can be written as affine constraints on the new moment $p_k$ in terms of lower-order moments:
\begin{align}
p_2 \;&\le\; 1, \label{eq:ineq-p2}\\
p_3 \;&\ge\; \frac{3p_2-1}{2}, \label{eq:ineq-p3}\\
p_4 \;&\le\; \frac{1-6p_2+3p_2^2+8p_3}{6}, \label{eq:ineq-p4}\\
p_5 \;&\ge\; \frac{-1+10p_2-15p_2^2-20p_3+20p_2p_3+30p_4}{24}. \label{eq:ineq-p5}
\end{align}
In general, the parity alternates: even $k$ yields an \emph{upper} bound on $p_k$, while odd $k$ yields a
\emph{lower} bound on $p_k$.

Since $p_1=1$ is simplified.

\section{Werner states and the one-parameter family}\label{app:werner-states}

A bipartite state $\rho\in\mathcal B(\mathbb C^d\otimes\mathbb C^d)$ is called a \emph{Werner state} if it is invariant
under conjugation by $U\otimes U$ for all unitaries $U\in U(d)$, i.e.
\begin{equation}
  (U\otimes U)\,\rho\,(U^\dagger\otimes U^\dagger)=\rho
  \qquad \forall\,U\in U(d).
\end{equation}
By this symmetry, any Werner state lies in the span of the identity $\mathbb I$ and the flip operator
\begin{equation}
  F \;=\; \sum_{i,j=0}^{d-1}\ket{i,j}\!\bra{j,i}.
\end{equation}
We use the following one-parameter normalization throughout:
\begin{equation}
  \label{eq:werner-def}
  \rho_W(t)=\frac{1}{d^2-dt}\bigl(\mathbb I-tF\bigr),\qquad t\in[-1,1].
\end{equation}
Then $\Tr(\rho_W(t))=1$. For $d\ge2$ and $t\in[-1,1]$ we have $d^2-dt>0$, and since
$\mathrm{spec}(F)\subset\{\pm1\}$ the operator $\mathbb I-tF$ has eigenvalues $1\pm t\ge0$.
Hence $\rho_W(t)\succeq 0$ on $t\in[-1,1]$.

\paragraph{Note}
In the qubit setting, we assume $N$ total qubits and a balanced bipartition $N/2\,|\,N/2$ (hence $N$ even).
Then each subsystem has local dimension $d=2^{N/2}$, so \eqref{eq:werner-def} applies on $\mathbb C^d\otimes\mathbb C^d$.

\subsection{Partial transpose and PT spectrum}

We now compute $\rho_W^{T_B}(t)$ and its spectrum. By linearity,
\begin{equation}
  \rho_W^{T_B}(t)=\frac{1}{d^2-dt}\bigl(\mathbb I-tF^{T_B}\bigr).
\end{equation}
Starting from $F=\sum_{i,j}\ket{i}\!\bra{j}\otimes\ket{j}\!\bra{i}$, partial transpose on $B$ yields
\begin{equation}
  F^{T_B}
  = \sum_{i,j}\ket{i}\!\bra{j}\otimes(\ket{j}\!\bra{i})^T
  = \sum_{i,j}\ket{i}\!\bra{j}\otimes\ket{i}\!\bra{j}
  = \sum_{i,j}\ket{i,i}\!\bra{j,j}.
\end{equation}
Introduce $\ket{\Phi_d}=\frac{1}{\sqrt d}\sum_{k=0}^{d-1}\ket{k,k}$ and $P=\ket{\Phi_d}\!\bra{\Phi_d}$.
Then $P=\frac{1}{d}\sum_{i,j}\ket{i,i}\!\bra{j,j}$, hence
\begin{equation}
  \label{eq:FPT}
  F^{T_B}=dP,
\end{equation}
and therefore
\begin{equation}
  \label{eq:werner-pt}
  \rho_W^{T_B}(t)=\frac{1}{d^2-dt}\bigl(\mathbb I-tdP\bigr).
\end{equation}

Since $P$ is rank one, $\mathbb I-tdP$ has eigenvalue $(1-td)$ along $\ket{\Phi_d}$ and eigenvalue $1$
on the orthogonal complement (dimension $d^2-1$). Dividing by $(d^2-dt)$ yields
\begin{equation}
  \label{eq:pt-eigs}
  \lambda_-=\frac{1-dt}{d^2-dt}\ \ (\text{mult. }1),\qquad
  \lambda_+=\frac{1}{d^2-dt}\ \ (\text{mult. }d^2-1).
\end{equation}

\paragraph{PPT threshold.}
For $d\ge 2$ and $t\in[-1,1]$, one has $d^2-dt\ge d(d-1)>0$, hence $\lambda_+>0$.
Therefore $\rho_W(t)$ is PPT iff $\lambda_-\ge 0$, i.e.
\begin{equation}
  \label{eq:ppt}
  \boxed{\rho_W(t)\ \text{is PPT} \iff t\le \frac1d.}
\end{equation}

\begin{figure}
    \centering
    \includegraphics[width=0.5\linewidth]{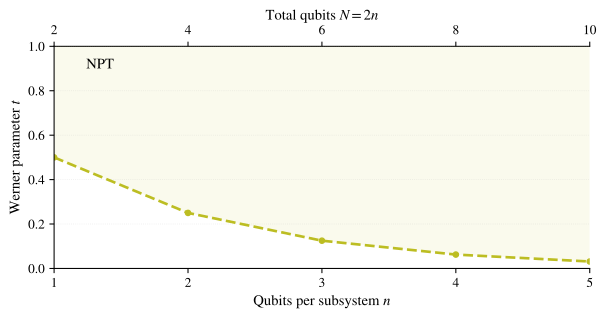}
    \caption{\textbf{PPT threshold for Werner states as a function of system size.}
    The threshold $t^* = 1/d = 2^{-N/2}$ decreases exponentially with the number
    of qubits $N$, making it harder to engineer NPT Werner states on larger systems
    while staying within $t\in[-1,1]$.}
\end{figure}
\subsection{Elementary symmetric polynomials and violation thresholds}

Let $\lambda_1,\dots,\lambda_{d^2}$ be the eigenvalues of $\rho_W^{T_B}(t)$ and define, for $k=0,1,\dots,d^2$,
\begin{equation}
  e_k
  :=\sum_{1\le i_1<\cdots<i_k\le d^2}\lambda_{i_1}\cdots\lambda_{i_k},
  \qquad e_0:=1.
\end{equation}
(For $k>d^2$, $e_k:=0$.) Our PT--moment inequalities are $e_k\ge 0$ (see Appendix~\ref{app:ng}).

With one eigenvalue $\lambda_-$ and $d^2-1$ copies of $\lambda_+$, we obtain for $k\ge 1$:
\begin{equation}
  \label{eq:ek-closed}
  e_k
  =
  \binom{d^2-1}{k}\lambda_+^k
  +\binom{d^2-1}{k-1}\lambda_-\,\lambda_+^{k-1}.
\end{equation}
Since $d^2-dt>0$ on our parameter range, $\lambda_+=1/(d^2-dt)>0$, so dividing by
$\lambda_+^{k-1}$ preserves inequalities.

Using $\lambda_-/\lambda_+=1-dt$ and $\binom{n}{k}/\binom{n}{k-1}=(n-k+1)/k$ with $n=d^2-1$ gives
\[
e_k<0
\iff
\frac{d^2}{k}-dt<0
\iff
t>\frac{d}{k}.
\]
Therefore,
\begin{equation}
  \label{eq:ek-threshold}
  e_k(\rho_W^{T_B}(t))<0
  \quad\Longleftrightarrow\quad
  t>\frac{d}{k}.
\end{equation}
In particular, on $t\in[-1,1]$ the condition $t>d/k$ can only be met when $k>d$
(since $d/k\ge1$ for $k\le d$). Thus the first possible violated order is $k=d+1$.
\begin{figure}
    \centering
    \includegraphics[width=1\linewidth]{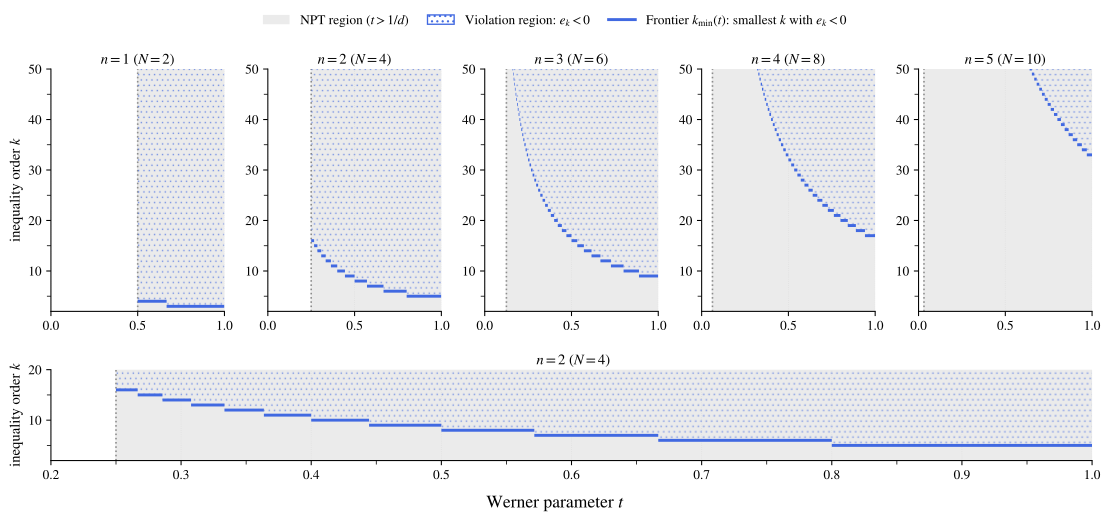}
    \caption{\textbf{Engineering window for Werner states.}
    Each horizontal band corresponds to the interval $t\in(d/k,\, d/(k-1)]$
    in which the $k$-th moment is the first to violate the PPT inequality,
    shown here for $N=2,4,6,8$ qubits ($d=2^{N/2}$).
    Inset: zoom into the high-$t$ region showing how the windows narrow
    as $k$ increases and $d/k\to 0$.}
\end{figure}
\paragraph{Engineering window.}
Fix $k>d$. The condition that $k$ is the \emph{smallest} violated order is $e_k<0$ and $e_j\ge 0$ for all $j<k$.
Since $e_j<0 \iff t>d/j$, it suffices to enforce $e_{k-1}\ge 0$ in addition to $e_k<0$, yielding
\begin{equation}
  \label{eq:window}
  \boxed{
  t\in\left(\frac{d}{k},\ \frac{d}{k-1}\right].
  }
\end{equation}

\begin{table}[t]
   \centering
   \begin{tabular}{lrrr}
        \hline
        State & Qubits & Parameter $t$ & First violated moment $k$ \\ \hline
        Werner & 2 & 0.8333 & 3 \\
        Werner & 2 & 0.5833 & 4 \\
        Werner & 4 & 0.9000 & 5 \\
        Werner & 4 & 0.7333 & 6 \\
        Werner & 6 & 0.8444 & 10 \\
        Werner & 6 & 0.9444 & 11 \\
        Werner & 8 & 0.9150 & 18 \\
        Werner & 8 & 0.9706 & 17 \\ \hline
   \end{tabular}
   \caption{Werner state benchmark instances. The mixing parameter $t$ is chosen so that entanglement is first detected at moment order $k$.}
   \label{tab:states-benchmarking}
\end{table}
\subsection{Further Results}\label{app:further-results}
We extend the results from the main paper in this section. \Cref{fig:2-qubit-5833-error} shows the relative errors of the estimators for the 2-qubit Werner state instance with $t=0.5833$, and \Cref{fig:4-qubit-9-error} visualizes the relative errors for the 4-qubit Werner state with $t=0.9$.

\begin{figure}[ht]
    \centering
    \includegraphics[width=1\linewidth]{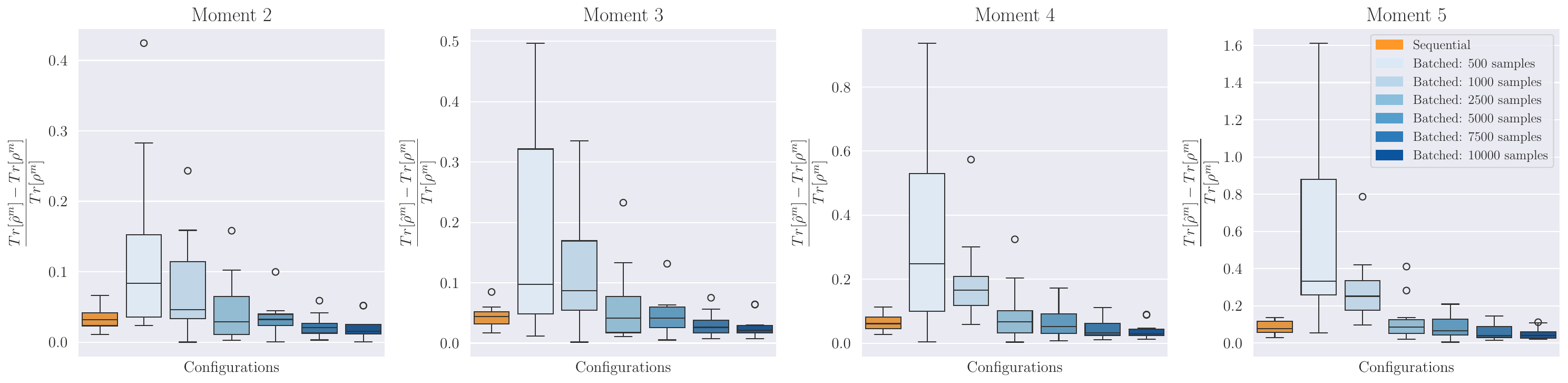}
    \caption{\textbf{2-qubit Werner State, $t=0.5833$:} The maximum number of samples used for the online estimate is 5,932.}
    \label{fig:2-qubit-5833-error}
\end{figure}

\begin{figure}[ht]
    \centering
    \includegraphics[width=0.75\linewidth]{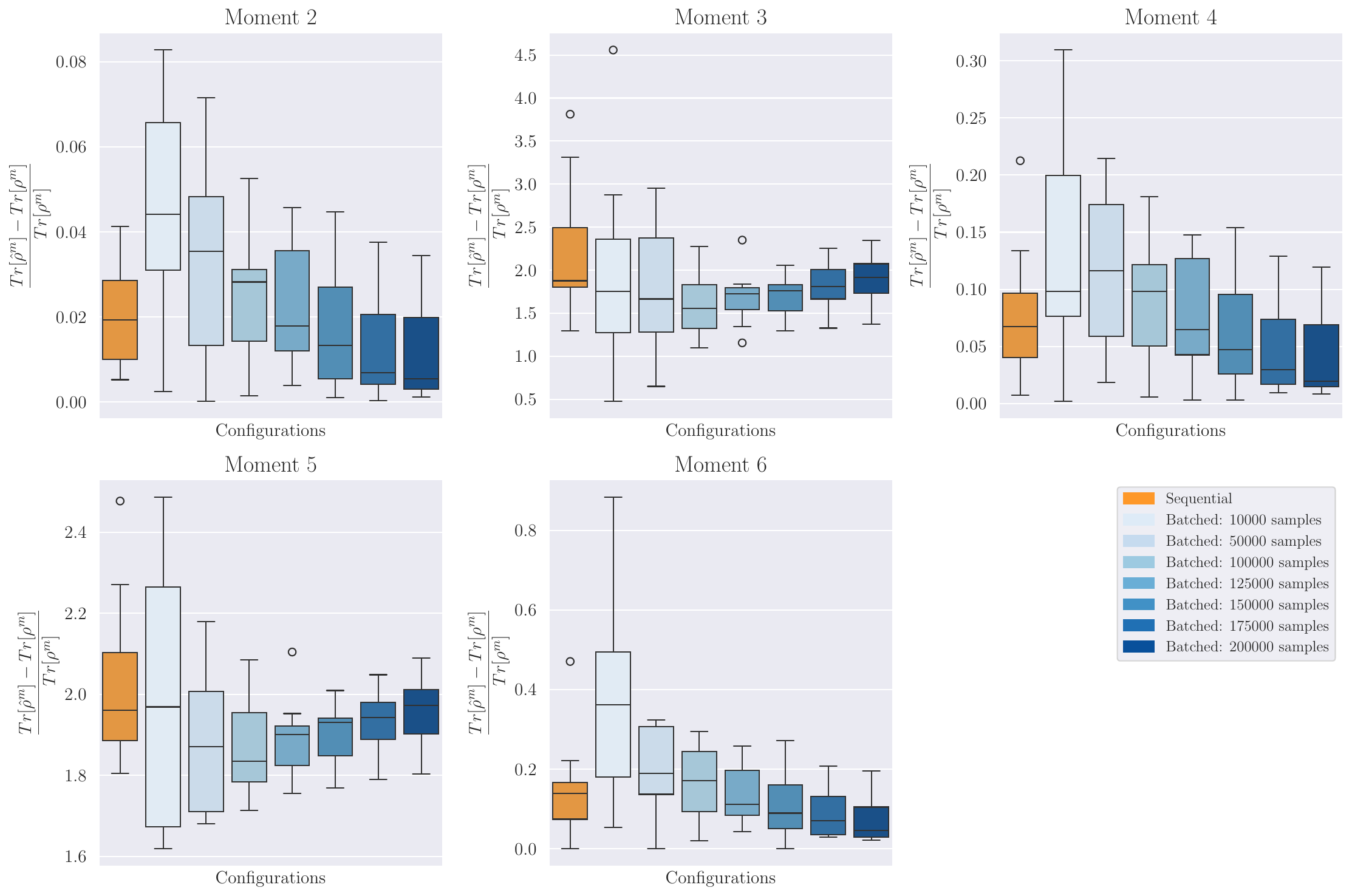}
    \caption{\textbf{4-qubit Werner State, $t=0.9$:} The maximum number of samples used for the online estimate is 160,828.}
    \label{fig:4-qubit-9-error}
\end{figure}

\end{document}